\documentclass[11pt]{article}

\usepackage{amsthm}
\usepackage{hyperref}
\usepackage{amssymb}
\usepackage{latexsym}
\usepackage{amsmath}
\usepackage{amsfonts}
\usepackage{mathrsfs} 

\usepackage{version}
\usepackage[dvips]{graphicx}
\usepackage[LGR,T1]{fontenc}%
\usepackage[french,english]{babel}
\usepackage{url}
\usepackage{enumerate}
\usepackage{hyperref}
\usepackage{color}
\usepackage{bm}
\usepackage{ascmac, fancybox}
\usepackage{caption}
\usepackage{subcaption}
\usepackage{float}
\usepackage{framed}
\usepackage{multirow}
\usepackage{arydshln}
\usepackage{adjustbox}
\usepackage{graphicx}
\usepackage{soul}

\numberwithin{equation}{section}
\theoremstyle{plain}
\newtheorem{thm}{Theorem}[section]
\newtheorem{prop}{Proposition}[section]
\newtheorem{lm}{Lemma}[section]
\newtheorem{claim}{Claim}
\numberwithin{claim}{section}
 
\numberwithin{cor}{section}
\newtheorem{example}{Example}[section] 

\theoremstyle{definition} 
\newtheorem{definition}{Definition} [section]

\numberwithin{conjecture}{section}
\theoremstyle{remark} 
\newtheorem{remark}{Remark}
\numberwithin{remark}{section}

\newcommand{\RomanNumeralCaps}[1]
    {\MakeUppercase{\romannumeral #1}}

\newcommand{\Var}{\mathrm{Var}}
\newcommand{\Cov}{\mathrm{Cov}}

\title{Systemic Risk Models for Disjoint and Overlapping Groups with Equilibrium Strategies}
\author{
\and  Yichen Feng\thanks{Department of Statistics and Applied Probability, University of California, Santa Barbara, CA 93106, USA (E-mail: \href{mailto:feng@pstat.ucsb.edu}{feng@pstat.ucsb.edu}).}
\and Jean-Pierre Fouque\thanks{Department of Statistics and Applied Probability, University of California, Santa Barbara, CA 93106, USA (E-mail: \href{mailto:fouque@pstat.ucsb.edu}{fouque@pstat.ucsb.edu}).}
  \and Ruimeng Hu\thanks{Department of Mathematics and Department of Statistics and Applied Probability, University of California, Santa Barbara, CA 93106, USA (E-mail: \href{mailto:rhu@ucsb.edu}{rhu@ucsb.edu}).}
\and Tomoyuki Ichiba\thanks{Department of Statistics and Applied Probability, University of California, Santa Barbara, CA 93106, USA (E-mail: \href{mailto:ichiba@pstat.ucsb.edu}{ichiba@pstat.ucsb.edu}).}
  }
\date{\today}

\begin{document}
\maketitle

\begin{abstract}

\end{abstract}
We analyze the systemic risk for disjoint and overlapping groups ({\it e.g.}, central clearing counterparties (CCP)) by proposing new models with realistic game features. Specifically, we generalize the systemic risk measure proposed in
[F. Biagini, J.-P. Fouque, M. Frittelli, and T. Meyer-Brandis, Finance and Stochastics, 24(2020), 513--564]
by allowing individual banks to choose their preferred groups instead of being assigned to certain groups. We introduce the concept of Nash equilibrium for these new models, and analyze the optimal solution under Gaussian distribution of the risk factor. We also provide an explicit solution for the risk allocation of the individual banks, and study the existence and uniqueness of Nash equilibrium both theoretically and numerically. The developed numerical algorithm can simulate scenarios of equilibrium, and we apply it to study the bank-CCP structure with real data and show the validity of the proposed model.

\section{Introduction}

Financial institutions are increasingly and tightly connected together at an unprecedented scale, and the complex dynamics of the inter-connectedness aggregate their idiosyncratic risks within the financial system. Consequently, failures of individual institutions due to excessive risk-taking may quickly propagate throughout the entire financial network and systemically cause cascading disasters. Such financial crises ({\it e.g.}, \cite{financial2011financial,haldane2011systemic,asmild2016controlling}) have dramatically demonstrated the importance of understanding the nature of systemic risk and designing models and methods to capture and analyze it. A large part of the current literature on systemic financial risk is concerned with the modeling structure of financial networks and the analysis of the contagion and the spread of a potential exogenous
shock into the system, {\it e.g.}, \cite{eisenberg2001systemic, gai2010liquidity, gai2010contagion, cifuentes2005liquidity, cherny2009new}. We refer interested readers to the book \cite{hurd2016contagion} for an exhaustive review. For a given financial network and a given random shock, one then determines the ``cascade'' mechanism, which generates many defaults. This mechanism often requires a detailed description of the balance sheet of each institution; assumptions on the interbank network and exposures, on the recovery rate at default, on the liquidation policy; the analysis of direct liabilities, bankruptcy costs, cross-holdings, leverage structures, fire sales, and liquidity freezes. Meanwhile, central clearing counterparties (CCPs) are introduced to the financial markets to mitigate the cascade defaults. They require default funds from their members to absorb the cost of clearing member's defauls, and several mechanisms of default funds have been designed \cite{albanese2020xva}.


In the previous works \cite{biagini2019unified, biagini2020fairness}, one of the authors (J.-P. Fouque) and his collaborators introduced a general class of systemic risk measures that allow for random allocations to individual banks before aggregation of their risks. They also proved the dual representation of a particular subclass of such systemic risk measures and the existence and uniqueness of the optimal allocation. They interpreted the systemic risk measures as the minimal amount of cash that secures the aggregated system by allocating capital to the single institutions before aggregating the individual risks, which allows for a possible ranking of the institutions in terms of systemic risk measured by the optimal allocations. \cite{armenti2018multivariate} developed an approach in a similar spirit, covering allocation first followed by aggregation frameworks.

In this paper, we generalize the systemic risk measure under exponential utility functions proposed in \cite{biagini2019unified, biagini2020fairness} by allowing individual banks to choose their allocations of the risks into different groups instead of being assigned to specific groups. This brings game features into modeling, making it more realistic and providing baselines for a CCP to design its default fund mechanism from rational banks. 
To solve the new models with game features, we first define the concept of Nash equilibrium based on banks' fair systemic risk allocation, a concept introduced in \cite{biagini2020fairness}, and discuss the existence and uniqueness of equilibrium strategies. Then some explicit formulas are derived under the Gaussian distribution of the risk factor. 
In the overlapping group case, we still focus on the risk measure produced by exponential utility functions and first provide explicit expressions for the systemic risk measure and fair risk allocation of each bank under general risk factor. Sensitivity and monotonicity properties are also established. The concept of Nash equilibrium is then extended to the overlapping group case, whose existence and uniqueness are discussed theoretically and numerically under Gaussian assumptions and with two groups. 
In particular, we propose a numerical algorithm based on fictitious play to identify the Nash equilibrium, and use it to study synthetic examples and the bank-CCP structure with real data, showing the validity of the proposed model.
 


The rest of the paper is organized as follows. In Section~\ref{sec:DisGroups}, we first describe the systemic risk models with game features for disjoint groups. Then we introduce the concept of group formation and Nash equilibrium, and analyze the optimal solution under Gaussian distribution of the risk factor. In Section~\ref{sec:Overlapping}, we generalize the model and concept of Nash equilibrium to the overlapping group case. We propose numerical methods for computing Nash equilibrium, and give several examples in Section~\ref{sec:algo}. We make conclusive remarks in Section~\ref{sec:conclusion}.

\section{Fair Systemic Risk Measure on Disjoint Groups}
\label{sec:DisGroups}
\subsection{Review on systemic risk measure}

A concept of systemic risk measure was proposed in \cite{biagini2019unified, biagini2020fairness}, where the authors considered the following problem of risk allocations of $\,N\,$ individuals labelled as $\,\{1, \ldots , N\}\,$. Given a joint distribution of an $\,N\,$-dimensional, real-valued random vector $\, {\mathbf X} \, :=\, (X^{1}, \ldots , X^{N})\,$ on a probability space $\, (\Omega, \mathcal F , \mathbb P) \,$, the risk sensitivity vector $\, {\bm \alpha} \, :=\, (\alpha^{1}, \ldots , \alpha^{N}) \in (0, \infty)^{N}\,$, the risk tolerance value $\, B < 0 \,$ and the partition set $\,\{I_{m} \, :=\, \{n_{m-1}+1 , \ldots , n_{m}\}, m \, =\, 1, \ldots , h\}\,$ (indexed by a vector $\,{\bm n} \, :=\, (n_{1}, \ldots , n_{h})\,$ with $\,0 \, =\, n_{0}<n_{1} < \cdots < n_{h} \, =\,  N\,$ for some $\,h \ge 1\,$) of the  $\,N\,$ elements $\,\{1, \ldots , N\}\,$, one defines the aggregate risk 
\begin{equation} \label{eq: AgRisk}
{\bm \rho} ( \mathbf X ) \, :=\,  \inf \Big \{ \sum_{n=1}^{N} Y^{n} : \mathbf Y \, =\, (Y^{1}, \ldots , Y^{N}) \in \mathcal C^{({\bf n})}_{0} , \mathbb E \Big[  \sum_{n=1}^{N} u_n ( X^{n} + Y^{n} ) \Big] \, =\,  B \Big \} \, 
\end{equation}
where we take $u_n = -\frac{1}{\alpha_n}e^{-\alpha_n\,x}$ as exponential utility functions, $\alpha_n$ could be interpreted as the risk aversion of individual $n$, and the random allocation $\, \mathcal C^{( \bf n )}_{0}\,$ of partition index vector $\, {\bm n} \,$ and the associated partition $\, \{ I_{m}, m \, =\,  1, \ldots , h\}\,$ is given by 
\begin{equation}
\begin{split}
\mathcal C^{( \bf n )}_{0} \, :=\,  \{ \mathbf Y \in \mathbb L^{0} ( \mathbb R^{N}) :  & \text{ there exists a real vector } {\bf d} \, :=\, (d_{1}, \ldots , d_{h}) \in \mathbb R^{h} \\
& \text{ such that } \sum_{i \in I_{m}} Y^{i} \, =\,  d_{m} \text{ for every } m \, =\,  1, \ldots , h \} \, . 
\end{split}
\end{equation}
The partition set $\, \{I_{m}, m \, =\, 1, \ldots  , h\}\,$ represents the grouping among the individuals and determined by the vector $\, {\bm n }\,$. 
Here $\, \mathcal C^{ ({\bm n})}_{0}\,$ is a subfamily of random vectors $\,\mathbb L^{0} (\mathbb R^{N}) \,$ associated with $\, {\bm n} \,$ such that all the partial sums of elements divided by the partition are deterministic real numbers. 

Theorem 6.2 in \cite{biagini2020fairness} shows that the infimum of \eqref{eq: AgRisk} is attained by 
\begin{equation}\label{eq: Thm-DisjointGroups}
\begin{split}
Y^{i}_{\mathbf X} \, :=\, - X^{i} + \frac{\, S_{m}  + d_{m}\,}{\,\alpha_{i} \, \beta_{m} \,}   \, , \quad S_{m} \, :=\,  \sum_{k\in I_{m}} X^{k} \, , \quad \beta_{m} \, :=\,  \sum_{k \in I_{m}} \frac{\,1\,}{\,\alpha_{k}\,} \, ,  \\
d_{m} \, :=\, \beta_{m} \log \Big( - \frac{\,\beta\,}{\,B\,} \mathbb E \Big[ \exp \Big ( - \frac{\,S_{m}\,}{\,\beta_{m}\,} \Big) \Big] \Big) \,, \quad \beta \, :=\, \sum_{n=1}^{N} \frac{\,1\,}{\,\alpha_{n}\,} \, =\,  \sum_{m=1}^{h} \beta_{m}, \, 
\end{split}
\end{equation}
for $\,m \, =\,  1, \ldots , h\,$ and $i \in I_m$; and 
\begin{equation}
{\bm \rho} ( \mathbf X ) \, =\, \sum_{n=1}^{N} Y^{n}_{\mathbf X} \, =\,  \sum_{m=1}^{h} d_{m} \, . 
\end{equation} 
Moreover, the systemic risk allocation $\,\rho^{i, ({\bm n})}(\mathbf X)\,$ of individual $\,i\,$ is given by 
\begin{equation} \label{eq: SysRAlloc}
\rho^{i, ({\bm n})}(\mathbf X) \, :=\, \mathbb E _{\mathbb Q^{m}_{\mathbf X}} [ Y^{i}_{\mathbf X} ] \, =\,  \Big( \mathbb E \Big[ \exp \Big( - \frac{\,S_{m}\,}{\,\beta_{m}\,}\Big) \Big]\Big)^{-1}  \mathbb E \Big [ Y^{i}_{\mathbf X}  \exp \Big( - \frac{\,S_{m}\,}{\,\beta_{m}\,}\Big) \Big]   \, ; \quad i \in I_{m},
\end{equation}
for $\, m \, =\,  1, \dots , h\,$, where $\, \mathbb Q_{\mathbf X}^{m}\,$ is a tilted probability measure, absolutely continuous with respect to $\,\mathbb P\,$, determined by the Radon-Nikodym derivative 
\begin{equation}\label{eq: QmX}
\frac{\,{\mathrm d} \mathbb Q_{\mathbf X}^{m}\,}{\,{\mathrm d} \mathbb P\,} \, :=\, \Big( \mathbb E \Big[ \exp \Big( - \frac{\,S_{m}\,}{\,\beta_{m}\,}\Big) \Big]\Big)^{-1}   \exp \Big( - \frac{\,S_{m}\,}{\,\beta_{m}\,}\Big) \, , \quad m \, =\,  1, \ldots , h \, . 
\end{equation}
By the construction, one has
\begin{equation}
 {\bm \rho} ( \mathbf X ) \, =\,  \sum_{m=1}^{h} d_{m} \, =\,  \sum_{m=1}^{h} \sum_{i\in I_{m}} \rho^{i, ({\bm n})} (\mathbf X ) \, =\, \sum_{n=1}^{N} \rho^{n, ({\bm n})} (\mathbf X)  \, . 
\end{equation} 

\subsection{Groups formation and Nash equilibrium}
In this section, we generalize the systemic risk measure to a game setup. For a game with $N$ individuals, we assume there are $N$ buckets $(B_1,\ldots, B_N)$ for each individual to choose which one she belongs to. The choice of individual $n$, $n=1,\ldots,N$, is denoted by $a_n$ and $a_n=j$ means individual $n$ chooses the bucket $j$. We call $\, \mathscr A\,$ the set of all strategies. A set of strategies ${\bm a}\,:=\, (a_1,\ldots, a_N)\in \mathscr A$ generates $m$ groups by considering only the non-empty buckets for some $1\leq m\leq N$.  Different sets of strategies may generate the same groups denoted by $\mathcal{C}({\bm a})$. We say two strategies $\, {\bm a}^{(1)} \,$ and $\, {\bm a}^{(2)}\,$ are equivalent, if the partitions $\, \mathcal C ({\bm a}^{(1)}) \,$ and $\, \mathcal C ( {\bm a}^{(2)}) \,$ are equivalent. With the individual systemic risk allocation in \eqref{eq: SysRAlloc}, the objective function of individual $n$ under the partition $\mathcal C ({\bm a})$ is defined by
\begin{align}
\rho^n(\mathcal{C}({\bm a})) \, =\, \rho^n (\mathbf X; \mathcal {C} ({\bm a})) \,=\,\mathbb{E}_{{\mathbb Q}_\mathbf{X}^m}[Y_\mathbf{X}^n] ,
\end{align}
where $n\in I_m$ for some $m\in \{1,\ldots,h\}$, $Y_\mathbf{X}^n$
 and ${\mathbb Q}_\mathbf{X}^m$ depend on the partition $\mathcal{C}(\bm a)$.
 
Let $\widehat{\bm a}=(\hat{a}_1,\ldots,\hat{a}_N)\in \mathscr A$ and $(\widehat{\bm a}^{-n},a^n)=(\hat{a}_1,\ldots,\hat{a}_{n-1},a^n,\hat{a}_{n+1},\ldots,\hat{a}_N)\in \mathscr A$.
\begin{definition}
With the systemic risk allocation map $\, \mathcal C \mapsto \rho^{\cdot} (\mathbf X; \mathcal C ) \,$ in \eqref{eq: SysRAlloc} and the above definitions, the configuration $\mathcal{C}(\widehat{\bm a})$ is a Nash Equilibrium if for every $n=\,  1, \ldots , N \,$ and $a^n$,
\begin{align}\label{eq: Nash net}
\rho^n(\mathbf X; \mathcal{C}(\widehat{\bm a}))\,\leq\,\rho^n(\mathbf X;\mathcal{C}(\widehat{\bm a}^{-n},a^n)), 
\end{align}
{\it i.e.}, the systemic risk allocation of individual $n$ is minimized under grouping $\mathcal{C}(\widehat{\bm a})$, given other individuals' choices are $\widehat {\bm a}^{-n}$. If there are multiple Nash equilibrium strategies  satisfying \eqref{eq: Nash net} and all the partitions associated with these Nash equilibrium strategies are equivalent, we say the Nash equilibrium strategy is unique up to equivalence relation. 
\end{definition}

In this paper, we shall consider the following questions: 
\begin{itemize}
\item Does a Nash equilibrium exist? 
\item If it exists, is it unique?
\end{itemize}

We view the equilibrium as a network of risk-sharing. It is easy to show that a single group with all the individuals, called full risk-sharing, is a Nash equilibrium. We call this the trivial Nash equilibrium. This follows from the fact that a configuration with a group having only one individual is never a Nash equilibrium; see Section 6.2 ``Monotonicity'' in \cite{biagini2020fairness}. For simplicity, we take Gaussian distribution for the risk factors and discuss the proposed model in detail with some examples.

\subsection{Extreme examples}

Here we discuss some extreme cases. Let $\,\lvert I_{m}\rvert ( \ge 1)\,$ be the number of elements in $\,I_{m}\,$ for $\,m \ge 1\,$. Under the exchangeability assumption on the joint distribution of $\, \mathbf X\,$ and the identical exponential utility functions with $\,\alpha_{1} \, =\,  \cdots \, =\,  \alpha_{N} \, =\,  \alpha > 0 \,$,  
the marginal distributions of $\, X^{i} e^{-S_{m}/\beta_{m}}\,$, $\,i \in I_{m}\,$ are identical with the same expectation, that is, 
\begin{equation} \label{eq: expectation of Xe}
\mathbb E \Big[ - {X^{i}} e^{-S_{m}/\beta_{m}} \Big] \, =\, \frac{\,- \beta_{m}\,}{\, \lvert I_{m}\rvert\,} \mathbb E \Big[ \frac{\, S_{m}\, }{\beta_{m}} e^{-S_{m}/\beta_{m}}  \Big] \, ; \quad i \in I_{m},
\end{equation}
and hence, with $\, \beta_{m} \, =\,  \alpha^{-1} \, \lvert I_{m} \rvert \, $, $\,\beta \, =\,  \alpha^{-1} \, N \,$ for $\,m \, =\,  1, \ldots , h \,$, 
\[
\mathbb E \Big[ Y^{i}_{\mathbf X} e^{-S_{m}/\beta_{m}} \Big] \, =\, \mathbb E \Big[ \Big( - X^{i} + \frac{\,S_{m}\,}{\,\alpha_{k} \beta_{m}\,} + \frac{\,d_{m}\,}{\,\alpha_{i} \beta_{m}\,} \Big) e^{-S_{m}/\beta_{m}}  \Big] 
\, =\,  \frac{\,d_{m}\,}{\, \lvert I_{m}\rvert\,} \mathbb E [ e^{-S_{m}/ \beta_{m}} ] \, . 
\]
Then substituting it into \eqref{eq: SysRAlloc} gives the systemic risk allocation of individual $i$ 
\[
\rho^{i} ( \mathbf X; \{I_{\cdot}\} ) \, =\, \frac{\,d_{m}\,}{\, \lvert I_{m}\rvert\,} \, =\, \frac{\,1\,}{\,\alpha\,} \log \Big(  \frac{\,-N\,}{\,\alpha \cdot (-B)\,} \mathbb E \Big[ e^{- \alpha S_{m}/ \lvert I_{m}\rvert\, }\Big] \Big)  \, ; \quad i \in I_{m} \, . 
\]


Under the i.i.d. Gaussian distribution assumptions for $\, \mathbf X\,$, {\it i.e.},  $\, X^{i}\,$, $\,i \, =\, 1, \ldots , N \,$ are independent, identically distributed Gaussian random variables with mean $\,\mu \in \mathbb R\,$ and variance $\, \sigma^{2} > 0  \,$, $\,S_{m} \, =\, \sum_{i\in I_{m}} X^{i}\,$ is distributed in normal with mean $\,\mu \lvert I_{m}\rvert\,$ and variance $\, \sigma^{2}\, \lvert I_{m}\rvert \,$ for $\, m \ge 1\,$. Direct calculation yields the systemic risk allocation of individual $i$ 
\begin{equation} \label{eq: exe Gauss}
\rho^{i} ( \mathbf X; \{I_{\cdot}\})  \, =\,  \frac{\,1\,}{\,\alpha\,} \log \Big( \frac{\,N \,}{\,\alpha \cdot (-B)\,} \Big) - \mu  +  \frac{\,\alpha \sigma^{2}\,}{\, 2 \lvert I_{m}\rvert\,} \, ; \quad i \in I_{m}, \, 
\end{equation}
is a decreasing function of the size $\, \lvert I_{m}\rvert\,$ of the group $\,I_{m}\,$ which individual $i$ belongs to. 

\begin{example}[I.I.D. Gaussian with the same exponential utility function] \label{eq: unique Nash eq1} The strategy $\, \widehat{\bm a} \, :=\, (1, \ldots , 1) \,$ that everyone chooses the same group is a unique Nash equilibrium.
In this case $\, \mathcal C ( \widehat{\bm a}) \, =\, \{ I_{1} \, =\,  \{ 1, \ldots , N\} \} \,$ with $\, \lvert I_{1}\rvert \, =\,  N\,$, $\, h \, =\,  1\,$. Every individual $\,k \,$ belongs to the same group $\,I_{1}\,$ and by \eqref{eq: exe Gauss}, 
\[
\rho^{k} ( \mathbf X \, ; \mathcal C( \widehat{\bm a} ))  \, =\,  \frac{\,1\,}{\,\alpha\,} \log \Big( \frac{\,N \,}{\,\alpha \cdot (-B)\,} \Big) - \mu  +  \frac{\,\alpha \sigma^{2}\,}{\, 2 N\,} \le \rho^{k} ( \mathbb X \, ; \mathcal C ( \widehat{\bm a}^{-k}, {a}_{k} ) ) \, ; \quad {\bm a} \in \mathscr A \, . 
\]
To see its uniqueness, if $\, {\bm a}^{\ast} \,$ were a Nash equilibrium strategy with $\, \mathcal C ( {\bm a}^{\ast}) \, =\,  \{ I_{m}^{\ast}, m \, =\,  1, \ldots , h^{\ast}\}\,$ which is not equivalent to $\, \mathcal C ( \widehat{\bm a }) \,$, then $\, h^{\ast} \ge 2\,$ and all the sets $\,  I_{m}^{\ast} \,$, $\, m \, =\,  1, \ldots , h^{\ast}\,$ satisfy $\, 1 \le  \lvert I_{m}^{\ast} \rvert \le N-1\,$. Take the group number $\,  {\ell}_{0} \, :=\, \text{arg}\min_{1\le i\le h} \lvert I_{i}^{\ast} \rvert \,$, of its smallest size. For a fixed individual $\,k \in I_{\ell_{0}}^{\ast}\,$, there exists $\, {\bm a} \in  \mathscr A\,$ with $\, {a}_{k} \, =\, j_{0}\,$ such that in the new partition $\,\mathcal C ( {\bm a}^{\ast -k}, {a}_{k}) \, =\,  \{ {I}_{m}\}\,$ the individual $\,k  \,$ belongs to another group $\,I_{m_{0}}\,$ with $\, \lvert I_{m_{0}}\rvert > \lvert I_{\ell_{0}}\rvert\,$, and hence by \eqref{eq: exe Gauss}, 
\[
\rho^{k} ( \mathbf X ; \mathcal C ( {\bm a}^{\ast}) ) \, =\,  \rho^{k}( \mathbf X ; \{ I^{\ast}_{m}, m \, =\, 1, \ldots , h^{\ast}\}) >  \rho^{k} ( \mathbf X; \mathcal C ( {\bm a}^{\ast -k }, j_{0})) \, . 
\]
This contradicts with the definition of Nash equilibrium. Thus, $\, \widehat{\bm a}\,$ is a unique Nash equilibrium up to equivalence relation. 
\end{example}

\begin{example}[Non-random, equal outcomes with the same exponential utility function] Instead, if $\, \mathbf X\,$ is a deterministic constant vector of $\,\mu (\in \mathbb R) \,$'s with $\, \sigma^{2} \equiv 0\,$, that is, $\, X^{i} \, =\,  \mu\,$ for every $\,i \, =\,  1, \ldots, N\,$, then there is no contribution from $\,I_{m}\,$ in the systemic risk allocation
\begin{equation}
\rho^{i} ( \mathbf X; \{I_{\cdot}\} ) \, =\,  \frac{\,1\,}{\,\alpha\,} \log \Big( \frac{\,N \,}{\,\alpha \cdot (-B)\,} \Big)  - \mu \, ; \quad i \in I_{m} \, , 
\end{equation} 
and hence, the risk sharing is arbitrary and undetermined. 
\end{example}



Next, we shall relax the condition on $\,\alpha\,$'s. We still assume $\, X^{i}\,$ are i.i.d. Gaussians with mean $\, \mu \in \mathbb R \,$ and variance $\, \sigma^{2} > 0 \,$. In this case, $\, S_{m} / \beta_{m}\,$ is normally distributed with mean $\, \mu \lvert I_{m}\rvert / \beta_{m}\,$ and variance $\, \lvert I_{m}\rvert \sigma^{2} / \beta_{m}^{2}\,$. Then direct calculations produces 
\begin{equation}
\begin{split} \label{eq: 2.6a}
\frac{d_{m}}{\, \beta_{m}\, } \, =\,  \log \Big( \frac{\,\beta\,}{\,-B\,} \Big) + \log \mathbb E \Big[ e^{-S_{m}/\beta_{m}}\Big]  \, =\,  \log \Big( \frac{\,\beta\,}{\,-B\,} \Big)  - \frac{\, \mu \lvert I_{m}\rvert\,}{\,\beta_{m}\,} + \frac{\,\lvert I_{m}\rvert \, \sigma^{2}\,}{\,2 \beta_{m}^{2}\,}  \, , 
\\
\mathbb E_{\mathbb Q^{m}_{\mathbb X}} \Big[ \frac{\,S_{m}\,}{\,\beta_{m}\,}\Big]  \, =\,  \Big( \mathbb E \Big[  e^{-S_{m}/\beta_{m}} \Big] \Big)^{-1}
\mathbb E \Big[  \frac{\,S_{m}\,}{\,\beta_{m}\,} e^{-S_{m}/\beta_{m}} \Big]
\, =\, \frac{\,\mu \, \lvert I_{m}\rvert\,}{\,\beta_{m}\,} - \frac{\, \lvert I_{m}\rvert\,\sigma^{2}\, }{\, \beta_{m}^{2}\,},
\end{split}
\end{equation}
where $\,\mathbb Q^{m}_{\mathbf X}\,$ is the tilted measure defined by \eqref{eq: QmX}. Also by \eqref{eq: expectation of Xe}, 
\begin{equation} \label{eq: 2.6b}
\mathbb E  \Big[ - X^{k} e^{-S_{m}/\beta_{m}} + \frac{\,S_{m}\,}{\,\alpha_{k} \beta_{m}\,} e^{-S_{m}/\beta_{m}}\Big] \, =\,  \Big( \frac{\,-\beta_{m}\,}{\, \lvert I_{m} \rvert \,} + \frac{\,1\,}{\,\alpha_{k}\,}\Big) \mathbb E \Big[  \frac{\,S_{m}\,}{\,\beta_{m}\,} e^{-S_{m}/\beta_{m}} \Big]\,,
\end{equation}
for $\, k \in I_{m}\,$. Hence, substituting \eqref{eq: 2.6a}-\eqref{eq: 2.6b} into \eqref{eq: SysRAlloc} brings, for $\, k \in I_{m}\,$,   
\begin{equation}
\begin{split}
\rho^{k}( \mathbf X ; \{I_{\cdot}\}) \, = &\,  \Big(  \frac{\,-\beta_{m}\,}{\, \lvert I_{m} \rvert \,} + \frac{\,1\,}{\,\alpha_{k}\,}\Big) \mathbb E_{\mathbb Q^{m}_{\mathbf X}} \Big[ \frac{\,S_{m}\,}{\,\beta_{m}\,} \Big]   + \frac{\,d_{m}\,}{\,\alpha_{k} \beta_{m}\,}\,  \\
\, = &\,  - \mu + \frac{\,1\,}{\,\alpha_{k}\,} \log \Big( \frac{\,\beta\,}{\,(-B) \,} \Big) + \frac{\,\sigma^{2}\,}{\,\beta_{m}\,} \Big( 1 - \frac{\, \lvert I_{m}\rvert\,}{\,2 \alpha_{k} \beta_{m}\,} \Big) \, . 
\end{split}
\end{equation} 
In order to find a Nash equilibrium, we evaluate the quantity 
\begin{equation} \label{eq: etak}
\eta^{k} (\{I_{\cdot}\}) \, :=\, \frac{\,1\,}{\,\beta_{m}\,} \Big( 1 - \frac{\, \lvert I_{m}\rvert\,}{\,2 \alpha_{k} \beta_{m}\,} \Big) \, =\,  \alpha_{k}  \frac{\,\prod_{\{j\in I_{m}:j\neq k\}} \alpha_{j}\,}{\,\sum_{j\in I_{m}} \alpha_{j}\,} \Big( 1 - \frac{\, \lvert I_{m}\rvert \, \prod_{\{j\in I_{m}:j\neq k\}} \alpha_{j}\,}{\, 2 \sum_{j\in I_{m}} \alpha_{j}\,} \Big) \, ; \quad k \in I_{m} \, .  
\end{equation}
Note that $\eta^{k}(\{I_{\cdot}\})$ is an increasing function of $\alpha_{k}$, and thus if $\text{arg}\max_{k \in I_{m}} \alpha_{k} = k_{0}$, then 
\begin{equation} \label{eq: obs1}
\eta^{k_{0}} ( \{I_{\cdot}\}) \ge \eta^{k} (\{ I_{\cdot}\}) \, ; \quad \text{ for every } \, \, k  \in I_{m} \, . 
\end{equation}
However, it is { not} necessarily true that $\, \rho^{k_{0}} ( \mathbf X ; \{I_{\cdot}\}) \ge \rho^{k} ( \mathbf X ; \{ I_{\cdot} \} ) \,$ for every $\, k \ge 0 \,$, because of the term $\, (1/\alpha_{k}) \log ( \beta / ( - B)) \,$. Note that $\, \eta^{k} ( \{ I_{\cdot} \}) \,$ is a decreasing function of $\, \lvert I_{m}\rvert /\beta_{m}\,$ and is a decreasing function of $\, \beta_{m}\,$.



\noindent  If $\,N \, =\, 2 \,$ and $\, 0 < \alpha_{1} \le \alpha_{2}\,$, then $\,\widehat{\bm a} \, :=\, (1,1) \,$ is a unique Nash equilibrium.
This is directly verified by the inequalities
\begin{equation} \label{eq: 2.9}
\eta^{k}(\mathcal C ( \widehat{\bm a})) \, =\,  \alpha_{k} \cdot \frac{\,\alpha_{1} \alpha_{2}\,}{\,(\alpha_{1} + \alpha_{2})^{2}\,} \le \frac{\,\alpha_{k}\,}{\,4\,} < \frac{\,\alpha_{k}\,}{\,2\,} \, =\,  \eta^{k} (\{\{1\}, \{2\} \}) \, ; \quad k \, =\,  1, 2 \,  ,
\end{equation}
and hence, $\, \rho^{k} ( \mathbf X;  \mathcal C ( \widehat{ \bm a}) ) < \rho^{k} ( \mathbf X  ; \mathcal C ( \widehat{\bm a }^{-k}, {\bf a} )) \,$ for every $\, k \, =\,  1, 2\,$, and $\, {\bm a} \in \mathscr A\,$ which is not equivalent to $\, \widehat{\bm a}\,$. 

\noindent If $\,N \ge 2\,$, and $\, 0 < \alpha_{1} \le \alpha_{2} \cdots \le \alpha_{N}\,$, then $\, \widehat{\bm a} \, :=\, (1, \ldots , 1) \,$ is a Nash equilibrium for a similar calculation to \eqref{eq: 2.9}, 
\[
\eta^{k} ( \mathcal C ( \widehat{\bm a} )) \, =\,  \alpha_{k} \cdot \frac{\,\prod_{i \neq k} \alpha_{i}\,}{\,\alpha_{1} + \cdots + \alpha_{k}\,} \Big( 1 - \frac{\, N \prod_{i\neq k} \alpha_{i}\,}{\,2  ( \alpha_{1} + \cdots + \alpha_{N})\,} \Big) 
<  \alpha_{k} \cdot \frac{\,\prod_{i \neq k} \alpha_{i}\,}{\,\alpha_{1} + \cdots + \alpha_{k}\,} \Big( 1 - \frac{\, \prod_{i\neq k} \alpha_{i}\,}{\, ( \alpha_{1} + \cdots + \alpha_{N})\,} \Big) 
\]
\[
\le \frac{\,\alpha_{k}\,}{\,4\,} < \frac{\,\alpha_{k}\,}{\,2\,} \, =\,  \eta^{k} ( \{\{k\} , \{1, \ldots , N\} \setminus \{k\} \} ) \, ; \quad k \, =\,  1, \ldots , N \, . 
\]
There is no reason to move out of the alliance of $\,\{1, \ldots , N\} \,$ and to become an outcast. This observation can be generalized: under this setup, there is no reason to move out of the alliance of a group $\, I_{m}\,$ of size greater than or equal to $\,2\,$, {\it i.e.}, $\, \lvert I_{m}\rvert \, \ge \, 2 \,$ and to become an outcast. 

\bigskip

We conjecture that $\,\widehat{\bm a}\,$  is unique Nash equilibrium 
under a wide range of configurations of $\, \alpha\,$'s. 
For each $\, \mathcal C ({\bm a}) \, =\,  \{ I_{m}, m \, =\, 1, \ldots , h\} \,$, let us consider the group heads $\,k_{m}\,$, $\, m \, =\,  1, \ldots , h\,$ and the head of the group heads $\, k^{\ast}\,$ by 
\begin{equation} \label{eq: heads}
k_{m} \, :=\,  \text{arg} \max_{\ell \in I_{m}} \alpha_{\ell} \, ; \quad m \, =\,  1, \ldots , h \, , \quad 
k^{\ast}\, :=\,  \text{arg} \!\!\!\!\!\!\!\!\max_{k \in \{k_{1}, \ldots , k_{m}\}} \eta^{k} ( \mathcal C ({\bf a})) \, . 
\end{equation}
Let us denote by $\, m^{\ast}\,$ the group name of $\,k^{\ast}\,$, {\it i.e.}, $\, k^{\ast} \in I_{m^{\ast}}\,$. By setting the group heads and the head of the group heads, we see from the observation made in \eqref{eq: obs1} that 
\begin{equation} \label{eq: headsheads1}
\eta^{k_{m}} (\mathcal C ({\bm a})) \ge \eta^{\ell} (\mathcal C ({\bm a})) \, ; \quad \text{ for every } \ell \in I_{m}\, , \quad m \, =\,  1, \ldots , h \, , 
\end{equation}
and 
\begin{equation} \label{eq: headsheads2}
\eta^{k^{\ast}} ( \mathcal C ( {\bm a })) \ge \eta^{k_{m}} ( \mathcal C ( { \bm a })) \, ; \quad \text{ for every } m \, =\,  1, \ldots , h \, . 
\end{equation}

\begin{lm} If there is a group head $\,k_{\ast} \in \{k_{1}, \ldots , k_{m}\} \setminus \{ k^{\ast}\} \,$ in a group $\,m_{\ast}\,$, i.e., $\,k_{\ast} \in I_{m_{\ast}}\,$ from $\, \mathcal C ( {\bm a}) \, =\,  \{I_{m}, m \, =\,  1, \ldots , h \} \,$ such that 
\begin{equation} \label{eq: cond1}
\frac{\,\alpha_{k^{\ast}}}{ \alpha_{k_{\ast}}} \, \le\, 1 + \frac{\,1\,}{\, \lvert I_{m_{\ast}}\rvert\,} \,, 
\end{equation}
then $\,\eta^{k^{\ast}}(\mathcal C( {\bm a})) \ge \eta^{k^{\ast}} ( \mathcal C ( \widetilde{\bm a})) \,$, where $\, \mathcal C ( \widetilde{\bm a}) \, =\,  \{ \widetilde{I}_{m}, m \, =\,  1, \ldots , h \}\,$ is obtained only by removing $\,k^{\ast}\,$ from group $\,m^{\ast}\,$ and adding $\,k^{\ast}\,$ into group $\,m_{\ast}\,$, that is, 
\[
\widetilde{I}_{m^{\ast}} \, :=\, I_{m^{\ast}} \setminus \{k^{\ast} \} \, , \quad \widetilde{I}_{m_{\ast}} \, :=\, I_{m_{\ast}} \cup \{ k^{\ast} \} \, . 
\]
In addition, either if the inequality in \eqref{eq: cond1} is strict or if the strict inequality $\,\eta^{k^{\ast}} (\mathcal C ( {\bm a})) > \eta^{k_{\ast}} ( \mathcal C ( {\bm a})) \,$ holds, then $\, \mathcal C({\bm a}) \,$ is not a Nash equilibrium. 
\end{lm}

\begin{proof} 

We rewrite $\,\eta^{k} ( \{I_{\cdot}\}) \, =\,  f( 1/\beta_{m} ; \lvert I_{m}\rvert / (2 \alpha_{k})) \,$ in \eqref{eq: etak} with a quadratic function $\, f(x; a) \, :=\,  x ( 1 - a x) \,$, $\, x > 0 \,$. Note that $\, f(0) \, =\,  0 \, =\,  f(1/a) \,$ and $\, f(x; a) \,$ is increasing in the interval $\, (0, 1/(2a)) \,$. For each group head $\,k_{1}, \ldots , k_{h}\,$, we have 
\[
\frac{\, \lvert I_{m} \rvert\,}{\,\alpha_{k_{m}}\,} \, \le  \, \sum_{k \in I_{m}} \frac{\,1\,}{\,\alpha_{k}\,} \, =\,  \beta_{m} \, \quad \text{ or } \quad \frac{\,1\,}{\,\beta_{m}\,}  \le \frac{\,\alpha_{k_{m}}\,}{\, \lvert I_{m}\rvert\,}; \quad m \, =\,  1, \ldots , h \, . 
\]
Thus, with $\, a \, =\,  \lvert I_{m}\rvert / (2 \alpha_{k_{m}}) \,$, we have $\, 1/\beta_{m} \le 1/ (2 a)\,$, and hence, $\, x \mapsto f(x; \lvert I_{m}\rvert / (2\alpha_{k})) \,$ is increasing in the interval $\, (0, 1/\beta_{m}) \,$ for $\, m \, =\,  1, \ldots , h\,$. 

By the definition of $\,k^{\ast}\,$ in \eqref{eq: heads} and \eqref{eq: headsheads2}, we have 
\begin{equation}
\begin{split}
\eta^{k^{\ast}} ( \mathcal C ( {\bm a})) \, &=\,  \frac{\,1\,}{\, \beta_{m^{\ast}}\,} \Big( 1 - \frac{\, \lvert I_{m^{\ast}}\rvert \,}{\,2 \alpha_{k^{\ast}} \beta_{m^{\ast}}\,} \Big) \, \\
& \ge \, \frac{\,1\,}{\,\beta_{m_{\ast}}\,} \Big( 1 - \frac{\, \lvert I_{m_{\ast}}\rvert \,}{\,2 \alpha_{k_{\ast}} \beta_{m_{\ast}}\,} \Big) \, =\,  f \Big( \frac{\,1\,}{\,\beta_{m_\ast}\,}; \frac{\, \lvert I_{m_{\ast}}\rvert\,}{\,2 \alpha_{k_{\ast}} \,}\Big)  \, =\,  \eta^{k_{\ast}} ( \mathcal C ( {\bm a})) \, . 
\end{split}
\end{equation}
Then by the monotonicity of $\, x\mapsto f ( x; \lvert I_{m_{\ast}}\rvert / ( 2 \alpha_{k_{\ast}} )) \,$ in the interval $\, (0, 1 / \beta_{m_{\ast}}) \,$, we have 
\begin{equation}
\begin{split}
\eta^{k^{\ast}} ( \mathcal C ( {\bm a})) & \ge  f \Big( \frac{\,1\,}{\,\beta_{m_\ast}\,}; \frac{\, \lvert I_{m_{\ast}}\rvert\,}{\,2 \alpha_{k_{\ast}} \,}\Big) \ge 
f \Big( \frac{\,1\,}{\,\beta_{m_\ast} + (1/\alpha_{k^{\ast}}) \,}; \frac{\, \lvert I_{m_{\ast}}\rvert\,}{\,2 \alpha_{k_{\ast}} \,}\Big) 
\\
 \, &=\,  \frac{\,1\,}{\, \beta_{m_{\ast}} + (1 / \alpha_{k^{\ast}}) \,} \Big( 1 - \frac{\,\lvert I_{m_{\ast}}\rvert\,}{\, 2 \alpha_{k_{\ast}} ( \beta_{m_{\ast}} + (1/\alpha_{k^{\ast}}) )\,} \Big)\, 
 \\
 & \ge \,  \frac{\,1\,}{\, \beta_{m_{\ast}} + (1 / \alpha_{k^{\ast}}) \,} \Big( 1 - \frac{\,\lvert I_{m_{\ast}}\rvert + 1 \,}{\, 2 \alpha_{k^{\ast}} ( \beta_{m_{\ast}} + (1/\alpha_{k^{\ast}}) )\,} \Big) \, =\,  \eta^{k^{\ast}} ( \mathcal C ( \widetilde{ \bm a} )) \, , 
\end{split}
\end{equation} 
where we used \eqref{eq: cond1} in the last inequality and $\, \mathcal C ( \widetilde{ \bm a} ) \,$ is obtained only by removing $\, k^{\ast}\,$ from the group $\, m^{\ast}\,$ and adding $\, k^{\ast}\,$ into the group $\, m_{\ast}\,$. Thus, for $\,k^{\ast}\,$  it is better to move from $\, m^{\ast}\,$ to $\,m_{\ast}\,$, and hence, $\, \mathcal C ( {\bm a}) \,$ is not a Nash equilibrium. 
\end{proof}

\begin{example}
Suppose that we have $\, \mathcal C ( {\bm a}) \, =\,  \{I_{m}\}_{m=1, 2 , 3} \,$, $\,I_{1} \, =\, \{1, 2\}\,$, $\,I_{2} \, =\,  \{3, 4, 5\}\,$, $\, I_{3} \, =\,  \{6, 7, 8, 9, 10\}\,$ with 
\[
\alpha_{1} \, =\,  \alpha_{2}\, =\,  2\, , \quad 
\alpha_{3} \, =\,  \alpha_{4} \, =\,  \alpha_{5} \, =\,  3 \, , \quad 
\alpha_{6} \, =\,  \alpha_{7}\, =\,  \alpha_{8} \, =\,  \alpha_{9}\, =\,  4 \, , \quad \alpha_{10} \, =\,  5 \, . 
\]
The condition \eqref{eq: cond1} holds with a strictly inequality, and $\, \mathcal C ( {\bm a}) \,$ is not a Nash equilibrium. 
\end{example}

\begin{example}
Suppose that we have $\, \mathcal C ( {\bm a}) \, =\,  \{I_{m}\}_{m=1, 2} \,$, $\,I_{1} \, =\, \{1, 2, 3\}\,$, $\,I_{2} \, =\,  \{4, 5, 6\}\,$  with 
\[
\alpha_{1} \, =\,  \alpha_{2}\, =\,  2\, , \quad 
\alpha_{3} \, =\, 4\, , 
\quad 
 \alpha_{4} \, =\,  \alpha_{5} \, =\, \alpha_{6}\, =\,  3 \, .  
 \]
The condition \eqref{eq: cond1} does not hold, however, $\, \mathcal C ( {\bm a}) \,$ is not a Nash equilibrium. 
\end{example}

\subsection{Case discussion: correlated Gaussian distribution}\label{section: DisjointCorGau}
In a system with $N$ individuals, we assume the joint distribution of $\mathbf{X}=(X^i,i=1,\ldots,N)^T$ follows a multivariate Gaussian distribution, that is, $\mathbf{X} \sim N({\bm \mu},{\bm \Sigma})$ where ${\bm \mu}\in \mathbb{R}^N$ and ${\bm \Sigma}\in \mathbb{R}^{N\times N}$ is positive semi-definite. The exponential utility functions for $N$ individuals have positive parameters ${\bm \alpha}=(\alpha_1,\ldots,\alpha_N).$

For every partition set $I_m$, $m=1,\ldots,h$, define a group vector $A_m \in \mathbb{R}^{1\times N}$ which consists of only $0$'s and $1$'s. For all the $j\in I_m$, the $j$-th element in $A_m$ is $1$, otherwise $0$. For example, in a 4-player system, if individuals 1 and 4 in group 1, and individuals 2 and 3 in group 2, the corresponding vectors for the two groups are $A_1=(1,0,0,1)$ and $A_2=(0,1,1,0)$. Then, following \eqref{eq: Thm-DisjointGroups}, we have
\begin{equation}\label{eqn: Sm distribution}
S_m = \sum\limits_{i\in I_m} X^i\,=\, A_m\mathbf{X}\,\sim\,N\big(A_m{\bm \mu},A_m{\bm \Sigma}A_m^T\big)\,\stackrel{def}{=}N(\mu_m^s, (\sigma_m^s)^2).
\end{equation}

\noindent The results in {Appendix}~\ref{proof: expResults} produce that, for $m=1,\ldots,h$ and for $k\in I_m$,
\[
d_m = \beta_m\log\Big(-\dfrac{\beta}{B}\mathbb{E}\big[\exp (-S_m/\beta_m) \big]\Big)=\beta_m\log\Big (\frac{\beta}{-B}\Big)-\,\mu_m^s\,+\,\frac{(\sigma_m^s)^2}{2\beta_m} ,
\]
and 
the systemic risk allocation of individual $k$ is given by
\begin{align}
\mathbb E _{\mathbb Q^{m}_{\mathbf X}} [ Y^{k}_{\mathbf X} ]
&=\mathbb E  \Big [ Y^{k}_{\mathbf X}\cdot \dfrac{d{\mathbb Q}^m_\mathbf{x}}{d\mathbb{P}} \Big ]=\dfrac{\mathbb E\big[(- X^{k} + \frac{\, 1\,}{\,\alpha_{k} \, \beta_{m}}S_m+\,\frac{\, 1\,}{\,\alpha_{k} \, \beta_{m}}d_m)\cdot e^{-S_m/\beta_m}\big]}{\mathbb{E}\big(e^{-S_m/\beta_m}\big)}\nonumber\\
&=-\mu_k+\frac{1}{\alpha_k}\log \Big(\frac{\beta}{-B}\Big)+\frac{1}{\beta_m}A_m{\bm \Sigma}_{[,k]} -\frac{(\sigma_m^s)^2}{2\beta_m^2\alpha_k} .\label{eq: generalsolution}
\end{align}

\begin{remark}[Effect of Mean] \label{rmk: allinone and subgroups}
From the above formula of the systemic risk allocation for individual $k$, we find that the mean of individual $k$ has no effect on her risk allocation no matter which group she belongs to. So in the following discussion, without loss of generality, we take all means to be the constant zero.

\end{remark}

\begin{remark}[Comparison between Trivial Grouping and Multi-Groups]
The total risk allocation for multiple groups ($h\geq 2$) is always greater than the total risk allocation for the trivial grouping ($h=1$). The proof of this statement will be given in {Appendix}~\ref{proof: compare-Trivial-Multi-Groups}, and we refer interested readers to read ``Monotonicity'' in \cite{biagini2020fairness} for general proof free of the distribution of risk factors.

\end{remark}

In the following, we present three concrete examples to help better understand on banks' rational choices under this fair risk allocation. 

\begin{claim}\label{claim: 3.1}
If $\mathbf{X}=(X^i,i=1, \ldots, N)^T$ has the same standard deviation $\sigma > 0$ and correlation coefficient $\rho\in [-1,1)$, and the utility parameters are identical, denoted by $\alpha$($>0$). Then there is only one trivial Nash equilibrium, that is, all individuals being in the same group.
\end{claim}

\begin{claim}\label{claim: const_block_rho}
In the case of $N=4$, we assume all the utility parameter $\alpha_i$s are the same and equal to 1, and all individuals have the same standard deviation denoted by $\sigma>0$. If the correlation matrix is a block matrix with uniform correlation coefficient, {\it i.e.}, the correlation matrix is given by $ \left(
\begin{array}{cccc}
1\,\,&\rho\,\,&0\,\,&0\\
\rho &1 &0 &0\\
0&0&1&\rho\\
0&0&\rho&1
\end{array}
\right)$
, we have the following conclusion about Nash equilibrium. 

\begin{itemize}
\item If $\rho\in [-\frac{3}{13},\,\frac{3}{8}]$, there is no nontrivial Nash equilibrium.
\item If $\rho\in [-1,-\frac{3}{13})$, grouping "\{1,2\}-\{3,4\}", {\it i.e.}, the first two and the second two individuals are in two different groups, is a nontrivial Nash equilibrium.
\item If $\rho\in (\frac{3}{8},1]$, grouping "\{1,3\}-\{2,4\}" and "\{1,4\}-\{2,3\}" are both nontrivial Nash equilibriums for the system..
\end{itemize}
\end{claim}
\noindent According to the claim, when the correlation is not strong, all individuals tend to be together to form a trivial Nash equilibrium. Otherwise, negatively correlated individuals tend to be in the same group while positively correlated individuals tend to be separate . 


\begin{claim}\label{clami:const_block_rhofor5}
In the case of $N=5$, we assume the standard deviation are uniform and the $\alpha_i$'s are 1. If the covariance matrix is $\mathbf{\Sigma}=\, \left(
\begin{array}{ccccc}
\sigma^2\,\,&\rho\sigma^2\,\,&0\,\,&0\,\,&0\\
\rho\sigma^2 &\sigma^2 &0 &0&0\\
0&0&\sigma^2&\rho\sigma^2&\rho \sigma^2\\
0&0&\rho\sigma^2&\sigma^2&\rho \sigma^2\\
0&0&\rho\sigma^2&\rho \sigma^2&\sigma^2
\end{array}
\right)$, which is of block form with positive standard deviation $\sigma$ and the correlation coefficient $\rho\in (-1,1)$. 

\begin{itemize}
	\item If $\rho \in (-1, -2/7]$, grouping   ``$\{1,2\}-\{3,4,5\}$'' is a Nash equilibrium.
	\item If $\rho \in (-2/7, 1)$, there is no nontrivial Nash equilibrium. 
	\item Grouping ``$\{1,2,3\}-\{4,5\}$'' can not be a non-trivial Nash for any value of $\rho$. 
\end{itemize}

\end{claim}

We give derivations of Claims~\ref{claim: 3.1}, \ref{claim: const_block_rho} and \ref{clami:const_block_rhofor5} in Appendix~\ref{proof: claim 3.1}, \ref{proof: claim 3.2} and \ref{proof: claim 3.3}, respectively.
The result further shows that individuals tend to stay with highly negatively correlated individuals to minimize the systemic risk if they exist. It is impossible for individuals to stay with correlated and uncorrelated individuals at the same time except for the trivial case when all individuals are together.
 
 \section{Systemic Risk Measure on Overlapping Groups}
 \label{sec:Overlapping}
To further study the systemic risk measure under exponential utility functions, we generalize the systemic risk allocation for $N$ individuals on disjoint groups in section \ref{sec:DisGroups} to the risk allocation for them on overlapping groups where they can choose multiple groups to allocate their risks. Assuming there are at most $h$ groups, the weighted risk factors for the $n^{th}$-individual assigned to multiple groups are labelled as $w_{n,k}X^n$, $k=1,\ldots,h$ with $\sum_{k=1}^h w_{n,k}=1$. In the weight $w_{n,k}$, the index $k$ refers to the group number the individual $n$ joins and the weight can be of any value between $0$ and $1$. If $w_{n,k}=0$ for some $k$, then we say the individual is not in the $k$-th group.  Therefore, we can extend the systemic risk measure given by \eqref{eq: AgRisk} to a general measure $\rho$ defined by

 \begin{align}
{\bm \rho}(\mathbf{X})\,:\,=\inf\left\{\sum\limits_{n=1}^N\,\sum\limits_{k=1}^h Y^{n,k}\,:\,\mathbf{Y}\in\mathcal{C}_0^{\text{new}},\mathbb{E}\left[\sum\limits_{n=1}^N\,\sum\limits_{k=1}^h u_n(w_{n,k}X^n+Y^{n,k})\right]\,=B\,
\right\},  \label{eq: GeAgRisk}
\end{align}
where we take $u_n(x)=-\frac{1}{\alpha_n}e^{-\alpha_n \,x}$ as exponential utility functions; $h$ is the maximum number of groups individuals can contribute to in total and it is a finite integer; and the random allocations $\mathcal{C}_0^{\text{new}}$ is given by
\begin{align}\label{def:cnew}
\mathcal{C}_0^{\text{new}}=\big\{\mathbf{Y}=(&Y^{i,j},1\leq i\leq N,\,1\leq j\leq h)\in L^0(\mathbb{R}^{N\times h})\,:\,
\exists\, d=(d_1,\ldots. d_h)\in \mathbb{R}^h,\nonumber\\
 &\sum\limits_{i=1}^N Y^{i,j}=d_j, \,\text{for } j=1,\ldots, h
\big\}.
\end{align}

\begin{remark}
Here $h$ is an integer fixed a priori, to eliminate the situation that an individual wants to split the risk $X^i$ into infinitely many groups. An alternative way is to impose a minimum value requirement for non-zero weights to avoid too many groups for an individual to participate in, denoted by $w^{\min}$. Then, naturally $h=\lfloor\frac{1}{w^{\min}} \rfloor\cdot N$. 
\end{remark}

\begin{remark}
The generalized system \eqref{eq: GeAgRisk}-\eqref{def:cnew} still meets the assumptions made in \cite{biagini2020fairness}. Because the measure on overlapping groups can be seen as the measure \eqref{eq: AgRisk} on disjoint groups with more individuals with weighted risk factors. Thus the existence and uniqueness of optimal allocation solution $\mathbf{Y}_\mathbf{X}$ of the primal problem \eqref{eq: GeAgRisk} is guaranteed, according to the the discussion in Section 4 of \cite{biagini2020fairness}.
\end{remark}

\noindent Given the grouping for all individuals, we define the family of sets 
\begin{equation}\label{eqn:setMI-2}
\{I_j\,:\,=\{i\in \mathbb{N}\,:\,w_{i,j}> 0, i=1,\ldots,N\},\,j=1,\ldots,h\}.
\end{equation}
\begin{thm}\label{thm:JointResult}
The optimal value of ${\bm \rho} ( \mathbf X )$ in \eqref{eq: GeAgRisk} is attained by
\begin{align}
d_j&=\beta_j\log\left(-\frac{\beta}{B}\mathbb{E}\left[e^{-\frac{S_j}{\beta_j}}\right] \right),
\\
Y^{i,j}_\mathbf{X} &=\left[-w_{i,j}X^{i}+\frac{1}{\alpha_i\beta_j}\left(S_j+d_j\right)\right] \mathbf{1}_{w_{i,j}>0},\label{eqn:individualY}
\end{align}
where $S_j=\sum\limits_{i=1}^N w_{i,j}X^i=\sum\limits_{i\in I_j}w_{i,j}X^i$, $\beta_j=\sum\limits_{i=1}^N\frac{1}{\alpha_i}\mathbf{1}_{w_{i,j}>0}=\sum\limits_{i\in I_j}\frac{1}{\alpha_i}$, for $j=1,\cdots, h$ and $i=1,\ldots, N$, $\beta
=\sum\limits_{j=1}^h\beta_j=\sum\limits_{j=1}^h \sum\limits_{i=1}^N\frac{1}{\alpha_i}\mathbf{1}_{w_{i,j}>0}$,  
and
\begin{equation}
{\bm \rho} ( \mathbf X ) \, =\,  \sum_{j=1}^{h}\sum_{i=1}^N Y_\mathbf{X}^{i,j}= \sum_{j=1}^{h} d_{j} \, . 
\end{equation} 
The systemic risk allocation for individual i is $\sum\limits_{j=1}^h\mathbb{E}_{{\mathbb Q}_{\mathbf X}^j}[Y_\mathbf{X}^{i,j}]$ with the density 
\[
\dfrac{\mathrm{d}{\mathbb Q}^j_{\mathbf X}}{\mathrm{d}\mathbb{P}}:=\dfrac{e^{-\frac{S_j}{\beta_j}}}{\mathbb{E}\left[e^{-\frac{S_j}{\beta_j}}\right]},\quad j=1,\ldots,h.
\]
\end{thm}



\noindent The proof of Theorem~\ref{thm:JointResult} is left to Appendix~\ref{appendix:JointResult}. According to the theorem, we define 
\begin{equation}\label{eqn: individualTotalRisk}
\rho^i(\mathbf{X})\,:=\,\mathbb{E}_{{\mathbb Q}_{\mathbf X}}\left[Y_\mathbf{X}^{i}\right]=\sum\limits_{j=1}^h\mathbb{E}_{{\mathbb Q}^j_{\mathbf X}}\left[Y_\mathbf{X}^{i,j}\right]=\sum\limits_{j=1}^h\mathbb{E}\,\left[Y_\mathbf{X}^{i,j}\cdot \dfrac{\mathrm{d}{\mathbb Q}^j_{\mathbf X}}{\mathrm{d}\mathbb{P}}\right],
\end{equation}
as the {total fair systemic risk allocation} for individual $i$.

\begin{remark}
Compared with the disjoint group case well discussed in Biagini et al. \cite{biagini2020fairness}, the model here can be seen as an extended disjoint group case, where we consider that one individual can be divided into several sub-individuals and join different groups. { Thus, for the measures $\{Q_{\mathbf X}^{i,j}, 1\leq i\leq N, 1\leq j\leq h\}$, we have $Q_{\mathbf X}^{i,j}=Q_{\mathbf X}^{l.k}:=Q_{\mathbf X}^m$ if $j,k\in I_m$ for group $m=1,\ldots,h$.}
\end{remark}

\subsection{Sensitivity analysis}
Based on the main theorem \ref{thm:JointResult}, we perform a sensitivity analysis by adding a perturbation on the risk factors. Consider the risk factors are given by $\mathbf{X}+\varepsilon\mathbf{Z}$ where $\varepsilon\in \mathbb{R}$ and $\, {\mathbf X} \, :=\, (X^{1}, \ldots , X^{N})\, ,\, {\mathbf Z} \, :=\, (Z^{1}, \ldots , Z^{N})\, $ on a probability space $\, (\Omega, \mathcal F , \mathbb P) \,$, we have for $w_{i,j}>0$,
\begin{align}
Y^{i,j}_{\mathbf{X}+\varepsilon\mathbf{Z}} &=-w_{i,j}(X^{i}+\varepsilon Z^{i})+\frac{1}{\alpha_i\beta_j}\left(S_j+\varepsilon S_j^Z\right)+\frac{1}{\alpha_i\beta_j}d_j^{\mathbf{X}+\varepsilon\mathbf{Z}},\\
d_j^{\mathbf{X}+\varepsilon\mathbf{Z}}&=\beta_j\log\left(-\frac{\beta}{B}\mathbb{E}\left[e^{-\frac{S_j+\varepsilon S_j^Z}{\beta_j}}\right] \right),
\end{align}
where $S_j = \sum\limits_{i\in I_j} w_{i,j} X^i$, $S_j^\mathbf{Z} = \sum\limits_{i\in I_j} w_{i,j} Z^i$ and $S_j^{\mathbf{X}+\varepsilon\mathbf{Z}}=S_j+\varepsilon S_j^\mathbf{Z}$.
\begin{prop}
\label{proposition: sensitivity analysis-main}
Let $\rho$ be the systemic risk measure in \eqref{eq: GeAgRisk}.
\begin{itemize}
\item Marginal risk contribution of group $j$:
\begin{equation}
    \frac{\partial}{\partial \varepsilon}d_j^{\mathbf{X}+\varepsilon\mathbf{Z}}\bigg\lvert_{\varepsilon=0}=\mathbb{E}_{\mathbb{Q}^j_{\mathbf{X}}}\left[-S_j^{\mathbf{Z}}\right],\quad j = 1,\ldots,h.
\end{equation}

\item Local causal responsibility for individual $i$ in group $j$:
\begin{equation}
    \frac{\partial}{\partial \varepsilon}\mathbb{E}_{\mathbb{Q}^j_{\mathbf{X}}}\left[Y^{i,j}_{\mathbf{X}+\varepsilon\mathbf{Z}}\right]\bigg\lvert_{\varepsilon=0}
     = \mathbb{E}_{\mathbb{Q}^j_{\mathbf{X}}}\left[-w_{i,j}Z^i\right],\quad i\in I_j.
\end{equation}

\item Marginal risk allocation for individual $i$ in group $j$: 
 \begin{align}
     \frac{\partial}{\partial \varepsilon}\mathbb{E}_{\mathbb{Q}^j_{\mathbf{X}+\varepsilon\mathbf{Z}}}\left[Y^{i,j}_{\mathbf{X}+\varepsilon\mathbf{Z}}\right]\bigg\lvert_{\varepsilon=0}
     &=\mathbb{E}_{\mathbb{Q}^j_{\mathbf{X}}}\left[-w_{i,j}Z^i\right]
     -\frac{1}{\beta_j}\Cov_{\mathbb{Q}^j_{\mathbf{X}}}\left(Y^{i,j}_{\mathbf{X}},S_j^{\mathbf{Z}}\right)\nonumber
     \\
     &=\mathbb{E}_{\mathbb{Q}^j_{\mathbf{X}}}\left[-w_{i,j}Z^i\right]
     +\frac{w_{i,j}}{\beta_j}\Cov_{\mathbb{Q}^j_{\mathbf{X}}}\left(X^i,S_j^{\mathbf{Z}}\right)-\frac{1}{\alpha_i\beta_j^2}\Cov_{\mathbb{Q}^j_{\mathbf{X}}}\left(S_j,S_j^{\mathbf{Z}}\right).\label{eqn:sensitivity-1}
 \end{align}
\end{itemize}
\end{prop}


\noindent We leave the proof of Proposition~\ref{proposition: sensitivity analysis-main} to Appendix~\ref{Appendix: sensitivity analysis-main}. Note that if we replace $\mathbb{Q}^j_{\mathbf{X}}$ with $\mathbb{P}$, none of the results above hold. To interpret these formulas, first we look at the first term in \eqref{eqn:sensitivity-1}, $\mathbb{E}_{\mathbb{Q}^j_{\mathbf{X}}}\left[-w_{i,j}Z^i\right]$. This term contains only the increment $Z^i$ in individual $i$ and thus is not a systemic contribution. Summing this term over all individuals in group $j$ gives
\begin{equation}\label{eqn: overall marginal risk allocation of group $j$}
\sum\limits_{i\in I_j}\frac{\partial}{\partial \varepsilon}\mathbb{E}_{\mathbb{Q}^j_{\mathbf{X}+\varepsilon\mathbf{Z}}}\left[Y^{i,j}_{\mathbf{X}+\varepsilon\mathbf{Z}}\right]\bigg\lvert_{\varepsilon=0} = \mathbb{E}_{\mathbb{Q}^j_{\mathbf{X}}}\left[-S_j^{\mathbf{Z}}\right] = \frac{\partial}{\partial \varepsilon}d_j^{\mathbf{X}+\varepsilon\mathbf{Z}}\bigg\lvert_{\varepsilon=0}.
\end{equation}
This shows the first term contributes to the marginal risk allocation of individual $i$ without any systemic influence. When $Z^i$ is positive, which means an increment is added, it results in a risk deduction, regardless of the relation to other individuals. When $\mathbf{Z}$ is deterministic, we can see, in \eqref{eqn:sensitivity-1}, the marginal risk allocation to individual $i$ in group $j$ is $\mathbb{E}_{\mathbb{Q}^j_{\mathbf{X}}}\left[-w_{i,j}Z^i\right]=-w_{i,j}Z^i$ and the covariance terms don't exist anymore.

To better study the effect of other terms in \eqref{eqn:sensitivity-1}, we take $\mathbf{Z}=Z^k\mathbf{e}_k$ where $k\neq i$. Then from \eqref{eqn:sensitivity-1} we obtain:
\begin{align}\label{eqn:sensitivity-special}
     \frac{\partial}{\partial \varepsilon}\mathbb{E}_{\mathbb{Q}^j_{\mathbf{X}+\varepsilon Z^k\mathbf{e}_k}}\left[Y^{i,j}_{\mathbf{X}+ \varepsilon Z^k\mathbf{e}_k}\right]\bigg\lvert_{\varepsilon=0}
     &=\frac{w_{i,j}}{\beta_j}\Cov_{\mathbb{Q}^j_{\mathbf{X}}}\left(X^i,Z^k\right)-\frac{1}{\alpha_i\beta_j^2}\Cov_{\mathbb{Q}^j_{\mathbf{X}}}\left(S_j,Z^k\right).
 \end{align}
 Supposing that $\frac{w_{i,j}}{\beta_j}\Cov_{\mathbb{Q}^j_{\mathbf{X}}}\left(X^i,Z^k\right)<0$, we look at the first term which relates to the covariance between $(X^i,Z^k)$. When they have a negative correlation under the systemic risk probability $\mathbb{Q}^j_{\mathbf{X}}$, the increase in individual $k$ will result in a decrease of the risk allocation for individual $i$. That means, individual $i$ takes advantage of the decrease of others. Since the overall marginal risk allocation of group $j$ doesn't change according to \eqref{eqn: overall marginal risk allocation of group $j$}, some other individuals in the group would pay for this advantage. This is related to the last term.
 
 The last term in \eqref{eqn:sensitivity-1} or \eqref{eqn:sensitivity-special} contains both the systemic contribution $-\frac{1}{\beta_j^2}\Cov_{\mathbb{Q}^j_{\mathbf{X}}}\left(S_j,Z^k\right)$ which only depends on the group $S_j$, and the the systemic relevance part $1/{\alpha_i}$ of individual $i$. The systemic component is distributed among the individuals based on $1/\alpha_i$. In addition, this term compensates for possible risk decrease in the second term of \eqref{eqn:sensitivity-1}, since the overall marginal risk allocation of group $j$ is fixed.
 

\begin{prop}(Sensitivity with respect to weights). For any $i,j$ such that $w_{i,j}>0$,
\label{proposition: sensitivity analysis}
\begin{align}
\dfrac{\partial \mathbb{E}_{\mathbb{Q}_\mathbf{X}^j}\,\left[Y_\mathbf{X}^{i,j}\right]}{\partial w_{i,j}}&=-\mathbb{E}_{{\mathbb Q}_\mathbf{X}^j}\,\left[X^i\right]-\dfrac{1}{\alpha_i\beta_j^2}\Cov_{{\mathbb Q}_\mathbf{X}^j}\left(X^i,S_j\right)+\dfrac{w_{i,j}}{\beta_j}\Var_{{\mathbb Q}_\mathbf{X}^j}\left(X^i\right) ;
\\
&=-\mathbb{E}_{{\mathbb Q}_\mathbf{X}^j}\left[X^i\right]-\dfrac{1}{\beta_j} \Cov_{{\mathbb Q}_\mathbf{X}^j}\left(X^i, \dfrac{1}{\alpha_i\beta_j} S_j-w_{i,j}X^i\right)\nonumber
\end{align}
\end{prop}

\noindent We give the proof in Appendix~\ref{Appendix: sensitivity analysis}.
 
\subsection{Monotonicity}
In a grouping set sequence $\{I_1,\ldots, I_h\}$, for some set $I_m$, assume there is a non-empty subset $I_{m'}$ of $I_m$ and for every $k\in I_{m'}$, assume the weight for risk factor $X^k$ is $w_{k, m'}\in (0,w_{k,m}]$. Then define $I_{m''}=\{k\in I_m\,:\, w_{k,m}-w_{k,m'}>0\}$ and the corresponding weights are $w_{k,m''}=w_{k,m}-w_{k,m'}$ for all $k\in I_{m''}$. Then there will be $h+1$ groups and the new grouping set sequence is $\{I_1,\ldots,I_{m'},I_{m''},\ldots, I_h\}$ while the weights structure are the same as before except those of groups $I_{m'}$ and $I_{m''}$. The optimal risk allocations under the new grouping of the primal problem coincide with $Y^{k,r}$, $k\in I_r$, for $r\neq m$. For $r=m$, $k\in I_{m'}$ or $I_{m''}$, we first know $w_{k,m'}\leq w_{k,m}$, $\,w_{k,m''}\leq w_{k,m}$ and we have the following.

\begin{prop}\label{prop: monotonicity-2}
Under the above setup, define $Y^{k,m}$, $k\in I_{m}$, the optimal allocation of group $m$ to the primal problem given $h$ groups. Define $Y^{k,m'}$, $k\in I_{m'}$ and $Y^{k,m''}$, $k\in I_{m''}$ the optimal allocations of groups $m'$ and $m''$ to the primal problem given $h+1$ groups, where $I_{m'}\in I_m$ and $I_{m''}\in I_m$. Then
\begin{equation}\label{eqn: monotonicity}
\mathbb{E}_{{\mathbb Q}_{\mathbf{X}}^m}\left[\sum\limits_{k\in I_{m'}}\dfrac{w_{k,m'}}{w_{k,m}}Y^{k,m}\right]\leq \eta_m'\log\left\{-\dfrac{\beta'}{B}\mathbb{E}\left[\exp\left(-\frac{1}{\eta_m'}\sum\limits_{k\in I_{m'}}w_{k,m'}X^k \right)\right]\right\}, 
\end{equation}
where $ \eta_m'=\sum\limits_{k\in I_{m'}}\dfrac{w_{k,m'}}{w_{k,m}}\dfrac{1}{\alpha_k}$.

Particularly,  if both $\sum\limits_{k\in I_{m'}}w_{k,m'}X^k$ and $\sum\limits_{k\in I_{m''}}w_{k,m''}X^k$ are nonnegative, it holds that 
\[
\mathbb{E}_{{\mathbb Q}^m_{\mathbf{X}}}\left[\sum\limits_{k\in I_{m'}}\dfrac{w_{k,m'}}{w_{k,m}}Y^{k,m}\right]\leq d_{m'},\quad \mathbb{E}_{{\mathbb Q}^m_{\mathbf{X}}}\left[\sum\limits_{k\in I_{m''}}\dfrac{w_{k,m''}}{w_{k,m}}Y^{k,m}\right]\leq d_{m''},
\]
thus 
\begin{align*}
\sum\limits_{k\in I_{m'}}\dfrac{w_{k,m'}}{w_{k,m}}Y^{k,m}+\sum\limits_{k\in I_{m''}}\dfrac{w_{k,m''}}{w_{k,m}}Y^{k,m}
&=\sum\limits_{k\in I_{m}}Y^{k,m}=d_m\leq d_{m'}+d_{m''},
\end{align*}
where $d_{m'}=\beta_m'\log\left\{-\dfrac{\beta'}{B}\mathbb{E}\left[\exp\left(-\frac{1}{\beta_m'}\sum\limits_{k\in I_{m'}}w_{k,m'}X^k \right)\right]\right\}$ and $d_{m''}$ are the total risks of groups ${m'}$ and ${m''}$.
It points out that each individual profits from this decrease (of groups) by avoiding being (group) alone ($I_m\sim I_{m'}\cup I_{m''}$). 

Also, when the system allows only disjoint grouping, {\it i.e.}, each individual joins groups with weight $1$, the inequality \eqref{eqn: monotonicity} implies the monotonicity result in \cite{biagini2020fairness}.
\end{prop}

\noindent The proof is left to the Appendix~\ref{appendix: monotonicity-2}.

\subsection{Generalized group formation and Nash equilibrium}
For a game with $N$ individuals and $h$ groups, similar as before, we assume there are $h$ buckets for each individual to choose which ones she belongs to and how much she puts. Therefore it induces a corresponding weight matrix for all individuals $${\bm W}=(w_{i,j})=\left(
\begin{array}{c} 
w_{1}\\
\vdots\\
w_i\\
\vdots\\
w_{N}
\end{array}
\right)\in \mathbb{R}^{N\times h},$$  defined as their strategies to distribute their risks, in order to minimize the individual total risk allocation. Each vector $w_i$ contains values of weights showing which groups individual $i$ belongs to and how much she wants to distribute the risk. So there is a natural constraint: $\sum_{j=1}^h w_{i,j}=1$ for every $i=1,\ldots, N$. Recall that the weight of individual $i$ in group $j$ is denoted by $w_{i,j}\in [0,1]$, and $w_{i,j}=0$ means individual $i$ is not in group $j$, $w_{i,j}=1$ means individual $i$ only joins group $j$. The case $w_{i,j}\in (0,1)$ means, besides group $j$, individual $i$ joins some other groups at the same time. Different sets of strategies may generate the same groups denoted by $\mathcal{C}({\bm W})$.

The objective function of individual $i$ is defined by:
\begin{align}
\rho^i(\mathcal{C}({\bm W}))\,:=\,\mathbb{E}_{{\mathbb Q}_\mathbf{X}}[Y_\mathbf{X}^i] =\sum\limits_{j=1}^h\mathbb{E}_{{\mathbb Q}^j_{\mathbf X}}\left[Y_\mathbf{X}^{i,j}\right] ,
\end{align}
where $i=1,\ldots,N$ and $\rho^i$ is the total fair systemic risk allocation for individual $i$ defined in \eqref{eqn: individualTotalRisk}.
 
Let $\widehat{{\bm W}}=(\widehat{w}_{i,j})=\left(
\begin{array}{c} 
\widehat{w}_{1}\\
\vdots\\
\widehat{w}_{i}\\
\vdots\\
\widehat{w}_{N}
\end{array}
\right)$ and $(\widehat{\bm W}^{-i},w^{i})=\left(
\begin{array}{c} 
\widehat{w}_{1}\\
\vdots\\
w^i\\
\vdots\\
\widehat{w}_{N}
\end{array}
\right)$ be the weight matrix $\widehat{{\bm W}}$ with the weight vector for individual $i$, {\it i.e.}, the $i$-th row, is replaced by a new vector $w^i$ whose elements sum up to 1.
\begin{definition}
The grouping $\mathcal{C}(\widehat{{\bm W}})$ defined by the weight matrix $\widehat{\bm W}$ is a Nash equilibrium if for every $i$ and any $w^i$,
\begin{align*}
\rho^i(\mathcal{C}(\mathbf{X}; \widehat{{\bm W}}))\,\leq\,\rho^i(\mathcal{C}(\mathbf{X}; \widehat{\bm W}^{-i},w^{i})),
\end{align*}
{\it i.e.}, the systemic risk allocation of individual $i$ is minimized under grouping $\mathcal{C}(\widehat{{\bm W}})$, given other individuals' strategies are $\widehat{\bm W}^{-i}$. 
\end{definition}

\noindent According the definition of Nash equilibrium, it is to be determined that when the grouping, determined by the matrix ${\bm W}$, is optimized and how individuals distribute their risks under Nash equilibrium.

\begin{remark}\label{remark: Necessary and Sufficient Condition for B}
One can not claim that in this overlapping group case, it is still true that a single group with all the individuals is a (trivial) Nash equilibrium. It follows from the proof in {Appendix.} \ref{appendix: Necessary and Sufficient Condition for B}. When $B$, the minimal level of expected utility, is small, it means the system has a high tolerance with respect to risks. Then individuals tend to split into different groups so there is no trivial Nash equilibrium. It can help explain why banks tend to join multiple central clearing counterparties (CCPs) to allocate their risks.
\end{remark}

\subsection{Case discussion: correlated Gaussian distribution}
In this section, we take Gaussian distribution for the risk factors for simplicity and discuss in detail. Similar to section \ref{section: DisjointCorGau}, assume the joint distribution of $\mathbf{X}=(X^i,i=1,\ldots,N)^T$ follows a multivariate Gaussian distribution, that is, $\mathbf{X} \sim N({\bm \mu},{\bm \Sigma})$ where ${\bm \mu}\in \mathbb{R}^N$ and ${\bm \Sigma}\in \mathbb{R}^{N \times N}$ is positive semi-definite. And define the column vector of $\bm W$ as a group vector given by
\[
A_j=\left\{
\begin{array}{ll}
0,\quad &i\not\in I_j\\
w_{i,j}, &i\in I_j,\quad i=1,\cdots,N
\end{array}
\right.\quad\in \mathbb{R}^{1\times N},\quad \text{for}\,\, j=1,\ldots, h.
\]
Then the group sum follows
\[
S_j = \sum\limits_{i\in I_j}w_{i,j} X^i\,=\, A_j\mathbf{X}\,\sim\,N\big(A_j{\bm \mu},A_j{\bm \Sigma}A_j^T\big)\,\stackrel{def}{=}N(\mu_j^s, (\sigma_j^s)^2),
\]
where 
\[\mu_j^s=\sum\limits_{k\in I_j} w_{k,j}\mu_k,\quad (\sigma_j^s)^2=\sum\limits_{m,k\in I_j} w_{m,j}w_{k,j}\sigma_{km}.\]

\noindent Using the results in Appendix \ref{proof: expResults}, we have for $j=1,\ldots, h$ and for $i\in I_j$:
\[
d_j = \beta_j\log\Big(-\dfrac{\beta}{B}\mathbb{E}\big[\exp (-S_j/\beta_j) \big]\Big)=\beta_j\log \Big (\frac{\beta}{-B} \Big)-\,\mu_j^s\,+\,\frac{(\sigma_j^s)^2}{2\beta_j} ,
\]
and 
the systemic risk allocation of individual $i$ in group $j$ is given by
\begin{align}
\mathbb E _{\mathbb Q^{j}_{\mathbf X}} [ Y^{i,j}_{\mathbf X} ]
&=\mathbb E  \Big [ Y^{i,j}_{\mathbf X}\cdot \dfrac{d{\mathbb Q}^j_\mathbf{x}}{d\mathbb{P}} \Big]=\dfrac{\mathbb E\big[(- w_{i,j}X^{i} + \frac{\, 1\,}{\,\alpha_{i} \, \beta_{j}}S_j+\,\frac{\, 1\,}{\,\alpha_{i} \, \beta_{j}}d_j)\cdot e^{-S_j/\beta_j}\big]}{\mathbb{E}\big(e^{-S_j/\beta_j}\big)}\nonumber
\\
&=-w_{i,j}\Big(\mu_i-\frac{1}{\beta_j}A_j{\bm \Sigma}_{[,i]}  \Big)+\frac{\, 1\,}{\,\alpha_{i} \, \beta_{j}}\Big(\mu_j^s-\frac{(\sigma_j^s)^2}{\beta_j}  \Big)+\frac{\, 1\,}{\,\alpha_{i} \, \beta_{j}}\Big(\beta_j\log \Big(\frac{\beta}{-B}\Big)-\,\mu_j^s\,+\,\frac{(\sigma_j^s)^2}{2\beta_j} \Big)\nonumber
\\
&=\frac{1}{\alpha_i}\log \Big(\frac{\beta}{-B}\Big)-w_{i,j}\mu_i+\frac{w_{i,j}}{\beta_j}A_j{\bm \Sigma}_{[,i]} -\frac{(\sigma_j^s)^2}{2\beta_j^2\alpha_i} ,\label{eq: joint-generalsolution}
\end{align}
where $A_j{\bm \Sigma}_{[,i]}=\sum\limits_{k\in I_j} w_{k,j}\sigma_{ki}$.

\subsubsection{Optimal weights in a general system}\label{section:Optimal Weights in a General System}
Continuing with the Gaussian distribution assumption, we use the formula \eqref{eq: joint-generalsolution} to first find the optimal weight vector $w^*$ for a given individual assuming the weight structure for other individuals is known. Then we search for Nash equilibrium numerically using an algorithm in section \ref{sec:algo}. That is, we minimize the total fair systemic risk allocation of individual $i$ defined in \eqref{eqn: individualTotalRisk} over the weight distributions $w_i = (w_{i,j},j=1,\ldots,N)$,
\begin{align*}
\min_{w_i} \mathbb{E}_{\mathbb{Q}_\mathbf{X}}\left[Y^{i}\right]&=\min_{w_i}\sum\limits_{j=1}^h \mathbb{E}_{\mathbb{Q}_\mathbf{X}^j}\left[Y^{i,j}\right]\\
&=\min_{w_{i,j}}\sum\limits_{j=1}^h\mathbf{1}_{w_{i,j}>0} \left[\frac{1}{\alpha_i}\log \Big(\frac{\beta}{-B}\Big)-w_{i,j}\mu_i+\frac{w_{i,j}}{\beta_j}A_j{\bm \Sigma}_{[,i]} -\frac{(\sigma_j^s)^2}{2\beta_j^2\alpha_i} \right],
\end{align*}
subject to $$ \sum\limits_{j=1}^hw_{i,j}\mathbf{1}_{w_{i,j}>0}=1.$$

\noindent For simplicity, we consider the problem when individual $i$ joins {at most two groups} to discuss the optimal weights.

\subsubsection{Risk allocation under at most two groups}\label{section: 2 groups}
When individual $i$ joins at most $2$ groups, there are three cases to discuss about the weights. Without loss of generality, we assume the weights are denoted by $(w_{i,1}, w_{i,2})$ for individual $i$, the weights for others are fixed and there is at least one other individual in group 1 and 2. Given that
\begin{align*}
 \mathbb{E}_{\mathbb{Q}_\mathbf{X}}\left[Y^{i}\right]=
 &\mathbf{1}_{w_{i,1}>0} \left[\frac{1}{\alpha_i}\log(\frac{\beta}{-B})-w_{i,1}\mu_i+\frac{w_{i,1}}{\beta_1}A_1{\bm \Sigma}_{[,i]} -\frac{(\sigma_1^s)^2}{2\beta_1^2\alpha_i} \right]\\
 &+\mathbf{1}_{w_{i,2}>0} \left[\frac{1}{\alpha_i}\log(\frac{\beta}{-B})-w_{i,2}\mu_i+\frac{w_{i,2}}{\beta_2}A_2{\bm \Sigma}_{[,i]} -\frac{(\sigma_1^s)^2}{2\beta_2^2\alpha_i} \right],
\end{align*}
where $\beta_j = \sum_{i\in \{k:w_{k,j}>0\}}\frac{1}{\alpha_i}$ and $\beta =\sum_{j=1}^h \beta_j$,
we discuss the following two boundary cases and one non-boundary case.
\begin{itemize}
\item Boundary case 1: $(w_{i,1}, w_{i,2})=(1,0)$, then
\begin{align}
\mathbb{E}_{\mathbb{Q}_\mathbf{X}}\left[Y^{i}\right]
&= \frac{1}{\alpha_i}\log \Big(\frac{\beta}{-B}\Big)-\mu_i+\frac{1}{\beta_1}A_1{\bm \Sigma}_{[,i]} -\frac{(\sigma_1^s)^2}{2\beta_1^2\alpha_i}  \label{eqn:corner-1}\\
&= \frac{1}{\alpha_i}\log \Big(\frac{\beta}{-B}\Big)-\mu_i+\frac{1}{\beta_1}\left(\sum\limits_{k=1,k\neq i}^N w_{k,1}\sigma_{ki}+\sigma_{ii}\right) \nonumber \\
&-\frac{1}{2\beta_1^2\alpha_i} 
\left(\sum\limits_{m,k=1,\neq i}^N w_{k,1}w_{m,1}\sigma_{km}+2\sum\limits_{k=1, k\neq i}^Nw_{k,1}\sigma_{ki}+\sigma_{ii}\right).\nonumber
\end{align}

\item Boundary case 2: $(w_{i,1}, w_{i,2})=(0,1)$, then
\begin{align}
\mathbb{E}_{\mathbb{Q}_\mathbf{X}}\left[Y^{i}\right]
&= \frac{1}{\alpha_i}\log \Big(\frac{\beta}{-B}\Big)-\mu_i+\frac{1}{\beta_2}A_2{\bm \Sigma}_{[,i]} -\frac{(\sigma_2^s)^2}{2\beta_2^2\alpha_i} \label{eqn:corner-2}\\
&= \frac{1}{\alpha_i}\log\Big(\frac{\beta}{-B}\Big)-\mu_i+\frac{1}{\beta_2}\left(\sum\limits_{k=1,k\neq i}^N w_{k,2}\sigma_{ki}+\sigma_{ii}\right)  \nonumber
\\
&-\frac{1}{2\beta_2^2\alpha_i} 
\left(\sum\limits_{m,k=1,\neq i}^N w_{k,2}w_{m,2}\sigma_{km}+2\sum\limits_{k=1,k\neq i}^Nw_{k,2}\sigma_{ki}+\sigma_{ii}\right).\nonumber
\end{align}
In the above formulas, $\beta_1 = \sum\limits_{k\neq i\,,\,w_{k,1}>0}\frac{1}{\alpha_k}+\frac{1}{\alpha_i}$, $\beta_2 = \sum\limits_{k\neq i,w_{k,2}>0}\frac{1}{\alpha_k}+\frac{1}{\alpha_i}$ and $$\beta = \sum\limits_{k\neq i\,,\,w_{k,1}>0}\frac{1}{\alpha_k}+\sum\limits_{k\neq i\,,\,w_{k,2}>0}\frac{1}{\alpha_k}+\frac{1}{\alpha_i}. $$

\item Non-boundary case: $(w_{i,1}, w_{i,2})=(w,1-w)$, while $0<w<1$. Then
\begin{align}
\mathbb{E}_{\mathbb{Q}_\mathbf{X}}\left[Y^{i}\right]
&= \frac{2}{\alpha_i}\log \Big(\frac{\beta'}{-B}\Big)-\mu_i+\frac{w}{\beta_1}A'_1{\bm \Sigma}_{[,i]}+\frac{1-w}{\beta_2}A_2'{\bm \Sigma}_{[,i]} -\frac{(\sigma_1')^2}{2\beta_1^2\alpha_i} -\frac{(\sigma_2')^2}{2\beta_2^2\alpha_i} \label{eqn:non-corner}\\
&= \frac{2}{\alpha_i}\log\Big(\frac{\beta'}{-B}\Big)-\mu_i\nonumber\\
&+\frac{w}{\beta_1}\left(\sum\limits_{k=1,k\neq i}^N w_{k,1}\sigma_{ki}+w\sigma_{ii}\right)+\frac{1-w}{\beta_2}\left(\sum\limits_{k=1,k\neq i}^N w_{k,2}\sigma_{ki}+(1-w)\sigma_{ii}\right) \nonumber \\
&-\frac{1}{2\beta_1^2\alpha_i} 
\left(\sum\limits_{m,k=1,\neq i}^N w_{k,1}w_{m,1}\sigma_{km}+2w\sum\limits_{k=1,k\neq i}^Nw_{k,1}\sigma_{ki}+w^2\sigma_{ii}\right)
\nonumber\\
&
-\frac{1}{2\beta_2^2\alpha_i} 
\left(\sum\limits_{m,k=1,\neq i}^N w_{k,2}w_{m,2}\sigma_{km}+2(1-w)\sum\limits_{k=1,k\neq i}^Nw_{k,2}\sigma_{ki}+(1-w)^2 \sigma_{ii}\right), \nonumber
\end{align}
where $\beta_1, \beta_2$ are the same as before but $\beta'=\beta_1+\beta_2 =  \beta+\frac{1}{\alpha_i}$.

\end{itemize}

\subsubsection{Risk allocation comparison between boundary and non-boundary cases}
\label{section: Compare Risk Allocation between Boundary cases and Non-Boundary case}
Here we compare the minimal risk allocation of non-boundary case \eqref{eqn:non-corner} with the risks of boundary cases \eqref{eqn:corner-1} and \eqref{eqn:corner-2}. First, we investigate the non-boundary case and prove $ \mathbb{E}_{\mathbb{Q}_\mathbf{X}}\left[Y^{i}\right]$ in \eqref{eqn:non-corner} is a quadratic function of $w$ and the minimal point is $w^*$ 
given by \eqref{eqn:w*}.
Taking partial derivative of \eqref{eqn:non-corner} gives 
\begin{align*}
\dfrac{\partial \mathbb{E}_{\mathbb{Q}_\mathbf{X}}\left[Y^{i}\right]}{\partial w}
&=\frac{1}{\beta_1}\sum\limits_{k=1,k\neq i}^N w_{k,1}\sigma_{ki}+\frac{2w}{\beta_1}\sigma_{ii}-\frac{1}{\beta_2}\sum\limits_{k=1,k\neq i}^N w_{k,2}\sigma_{ki}-\frac{2(1-w)}{\beta_2}\sigma_{ii} \nonumber \\
&-\frac{1}{\beta_1^2\alpha_i} 
\left(\sum\limits_{k=1,\neq i}^Nw_{k,1}\sigma_{ki}+w\sigma_{ii}\right)
+\frac{1}{\beta_2^2\alpha_i} 
\left(\sum\limits_{k=1,\neq i}^Nw_{k,2}\sigma_{ki}+(1-w) \sigma_{ii}\right)\nonumber
\\
\\
&=\sum\limits_{k=1,k\neq i}^N \left(\frac{w_{k,1}}{\beta_1}-\frac{w_{k,2}}{\beta_2}\right)\sigma_{ki}-\sum\limits_{k=1,k\neq i}^N \left(\frac{w_{k,1}}{\beta_1^2\alpha_i}-\frac{w_{k,2}}{\beta_2^2\alpha_i}\right)\sigma_{ki}
\\
&+w\left[(\frac{2}{\beta_1}+\frac{2}{\beta_2})-(\frac{1}{\beta_1^2\alpha_i}+\frac{1}{\beta_2^2\alpha_i})\right]\sigma_{ii}-(\frac{2}{\beta_2}-\frac{1}{\beta_2^2\alpha_i})\sigma_{ii}=0,
\end{align*}
and hence, 
\begin{equation}\label{eqn:w*}
w^*=\dfrac{\sum\limits_{k=1,k\neq i}^N \left(\frac{w_{k,1}}{\beta_1^2\alpha_i}-\frac{w_{k,2}}{\beta_2^2\alpha_i}\right)\sigma_{ki}-\sum\limits_{k=1,k\neq i}^N \left(\frac{w_{k,1}}{\beta_1}-\frac{w_{k,2}}{\beta_2}\right)\sigma_{ki}+(\frac{2}{\beta_2}-\frac{1}{\beta_2^2\alpha_i})\sigma_{ii}}{\left[(\frac{2}{\beta_1}+\frac{2}{\beta_2})-(\frac{1}{\beta_1^2\alpha_i}+\frac{1}{\beta_2^2\alpha_i})\right]\sigma_{ii}}
\end{equation}
is a {minimizer}, since $$\dfrac{\partial^2 \mathbb{E}_{\mathbb{Q}_\mathbf{X}}\left[Y^{i}\right]}{\partial w^2}=\left[(\frac{2}{\beta_1}+\frac{2}{\beta_2})-(\frac{1}{\beta_1^2\alpha_i}+\frac{1}{\beta_2^2\alpha_i})\right]\sigma_{ii}>0.$$

\noindent Here, note that $\frac{2}{\beta_1}-\frac{1}{\beta_1^2\alpha_i}>0,\frac{2}{\beta_2}-\frac{1}{\beta_2^2\alpha_i}>0$, it is clear that the denominator of \eqref{eqn:w*} is positive. If
\[
\sum\limits_{k=1,k\neq i}^N \left(\frac{w_{k,1}}{\beta_1^2\alpha_i}-\frac{w_{k,2}}{\beta_2^2\alpha_i}\right)\sigma_{ki}-\sum\limits_{k=1,k\neq i}^N \left(\frac{w_{k,1}}{\beta_1}-\frac{w_{k,2}}{\beta_2}\right)\sigma_{ki}>-\Big(\frac{2}{\beta_2}-\frac{1}{\beta_2^2\alpha_i}\Big)\sigma_{ii},
\]
we have $w^*>0;$
and if 
\[
\sum\limits_{k=1,k\neq i}^N \left(\frac{w_{k,1}}{\beta_1^2\alpha_i}-\frac{w_{k,2}}{\beta_2^2\alpha_i}\right)\sigma_{ki}-\sum\limits_{k=1,k\neq i}^N \left(\frac{w_{k,1}}{\beta_1}-\frac{w_{k,2}}{\beta_2}\right)\sigma_{ki}<\Big(\frac{2}{\beta_1}-\frac{1}{\beta_1^2\alpha_i}\Big)\sigma_{ii},
\]
we have $w^*<1.$

In conclusion, let 
\[
A:=\sum\limits_{k=1,k\neq i}^N \left(\frac{w_{k,1}}{\beta_1^2\alpha_i}-\frac{w_{k,2}}{\beta_2^2\alpha_i}\right)\sigma_{ki}-\sum\limits_{k=1,k\neq i}^N \left(\frac{w_{k,1}}{\beta_1}-\frac{w_{k,2}}{\beta_2}\right)\sigma_{ki},
\]
\[
B_1:= \Big(\frac{2}{\beta_1}-\frac{1}{\beta_1^2\alpha_i}\Big)\sigma_{ii},\quad\quad 
B_2:=\Big(\frac{2}{\beta_2}-\frac{1}{\beta_2^2\alpha_i}\Big)\sigma_{ii}.
\]
In non-boundary case, the {local optimal weights} for individual $i$ joining two groups are {non-zero}, {\it i.e.}, $$w^*=\dfrac{A+B_2}{B_1+B_2}\in (0,1)\quad\text{and}\quad 1-w^*\in (0,1)$$ if and only if $-B _2<A<B_1$, {\it i.e.},
\begin{equation}\label{eqn:w^*condition}
-(\frac{2}{\beta_2}-\frac{1}{\beta_2^2\alpha_i})\sigma_{ii}<\,\sum\limits_{k=1,k\neq i}^N \left(\frac{w_{k,1}}{\beta_1^2\alpha_i}-\frac{w_{k,2}}{\beta_2^2\alpha_i}\right)\sigma_{ki}-\sum\limits_{k=1,k\neq i}^N \left(\frac{w_{k,1}}{\beta_1}-\frac{w_{k,2}}{\beta_2}\right)\sigma_{ki}\,<\Big(\frac{2}{\beta_1}-\frac{1}{\beta_1^2\alpha_i}\Big)\sigma_{ii}.
\end{equation}
In {Appendix}~\ref{appendix: Sufficient Condition for Local Optimal Weights}, we investigate the condition further by reducing it to a simplified sufficient condition.

\bigskip

Then given $w^*$, the minimal risk of \eqref{eqn:non-corner} is
\begin{align}\label{eqn:non-corner optimal-risk}
\mathbb{E}_{\mathbb{Q}_\mathbf{X}}\left[Y^{i}\right]\bigg|_{w=w^*}
&=\frac{2}{\alpha_i}\log \Big(\frac{\beta'}{-B}\Big)-\mu_i\\
&+w^* \sum\limits_{k=1,k\neq i}^N \left[\left(\frac{w_{k,1}}{\beta_1}-\frac{w_{k,2}}{\beta_2}\right)-\left(\frac{w_{k,1}}{\beta_1^2\alpha_i}-\frac{w_{k,2}}{\beta_2^2\alpha_i}\right)\right]\sigma_{ki}
+ \sum\limits_{k=1,k\neq i}^N \left(\frac{w_{k,2}}{\beta_2}-\frac{w_{k,2}}{\beta_2^2\alpha_i}\right)\sigma_{ki}\nonumber
\\
&+\left(\frac{(w^*)^2}{\beta_1}-\frac{(w^*)^2}{2\beta_1^2\alpha_i}\right)\sigma_{ii}+\left(\frac{(1-w^*)^2}{\beta_2}-\frac{(1-w^*)^2}{2\beta_2^2\alpha_i}\right)\sigma_{ii}\nonumber
\\
&-\frac{1}{2\beta_1^2\alpha_i} 
\sum\limits_{m,k=1,\neq i}^N w_{k,1}w_{m,1}\sigma_{km}
-\frac{1}{2\beta_2^2\alpha_i} 
\sum\limits_{m,k=1,\neq i}^N w_{k,2}w_{m,2}\sigma_{km}.\nonumber
\end{align}
\begin{itemize}
    \item When the condition \eqref{eqn:w^*condition} holds, {\it i.e.}, $w^*\in (0,1)$, we compare the minimal risk of non-boundary case \eqref{eqn:non-corner} with that of boundary cases \eqref{eqn:corner-1} and \eqref{eqn:corner-2} in {Appendix}~\ref{appendix: Necessary and Sufficient Condition for Optimal Weights}. We conclude that when it holds that
\begin{equation}\label{eqn:condition}
\left\{
\begin{array}{ll}
\frac{2}{\alpha_i}\log(\frac{\beta'}{-B})-\frac{1}{\alpha_i}\log(\frac{\beta}{-B})
&<\dfrac{(A-B_1)^2}{2(B_1+B_2)}
+\frac{1}{2\beta_2^2\alpha_i} 
\sum\limits_{m,k=1,\neq i}^N w_{k,2}w_{m,2}\sigma_{km},
\\
\frac{2}{\alpha_i}\log(\frac{\beta'}{-B})-\frac{1}{\alpha_i}\log(\frac{\beta}{-B})
&<\dfrac{(A+B_2)^2}{2(B_1+B_2)}
+\frac{1}{2\beta_1^2\alpha_i} 
\sum\limits_{m,k=1,\neq i}^N w_{k,1}w_{m,1}\sigma_{km},
\end{array}
\right.
\end{equation}
the minimal risk for individual $i$ is achieved at $(w_{i1},w_{i2})=(w^*,1-w^*)$ with non-zero weights. This is a necessary and sufficient condition to determine which one is superior when the condition \eqref{eqn:w^*condition} is true.

\item When the condition \eqref{eqn:w^*condition} doesn't hold, if $w^* \leq 0$ and 
\[
\frac{2}{\alpha_i}\log \Big(\frac{\beta'}{-B}\Big)-\frac{1}{\alpha_i}\log\Big(\frac{\beta}{-B}\Big)
\geq\frac{1}{2\beta_1^2\alpha_i} 
\sum\limits_{m,k=1,\neq i}^N w_{k,1}w_{m,1}\sigma_{km},
\]
the minimal risk is achieved at the boundary case $(w_{i1},w_{i2})=(0,1)$. If $w^* \geq 1$ and 
\[
\frac{2}{\alpha_i}\log \Big(\frac{\beta'}{-B}\Big)-\frac{1}{\alpha_i}\log(\frac{\beta}{-B})
\geq \frac{1}{2\beta_2^2\alpha_i} 
\sum\limits_{m,k=1,\neq i}^N w_{k,2}w_{m,2}\sigma_{km},
\]
the minimal risk is achieved at $(w_{i1},w_{i2})=(1,0)$.
\end{itemize}

\section{Numerical Algorithm for Nash Equilibrium}
\label{sec:algo}
For a game of (large) $N$ individuals, it is hard to find the risk allocation of each individual associated with the grouping system. Instead, using the discussion introduced in Section \ref{sec:DisGroups} and Section \ref{section:Optimal Weights in a General System}, {under the assumption of Gaussian distribution for risk factors,} we can do numerical analysis of some examples via Python to search for Nash equilibrium for the system, such that no individual could achieve a smaller fair risk allocation by changing grouping or weights under the equilibrium. We conclude that, for the disjoint group case, non-trivial Nash equilibrium does not always exist and, neither does the overlapping group case. If we apply the overlapping group setup to the real world, {that is, interpreting $\sum_{n=1}^N{Y^{n,j}}$ as the default fund of the CCP $j$ that is liable for any participating institution/bank}, 
the numerical results indicate that big banks tend to join multiple CCPs while small banks tend to choose one.

\subsection{Numerical algorithm}\label{sec:algorithm}
In this section, we introduce numerical algorithms based on fictitious play. For the disjoint group case discussed in section \ref{sec:DisGroups}:
\begin{enumerate}
	\item Let $N$ individuals be in $N$ different groups, {\it i.e.},  $a_n = n$ for $1\leq n \leq N$, as the initial state;
	\item  At each stage, one individual is randomly picked with equal probability and it chooses to join the group which gives the minimal risk allocation. 
	\item Step 2 is repeated until the grouping is stabilized, and no individual has the incentive to move anymore.
\end{enumerate}

For the overlapping group case in section \ref{sec:Overlapping}, the algorithm is similar. We take the number of groups $h$ fixed and the initial weights for every individual in groups are randomly generated. The optimal weights at each stage are determined based on the discussion in section \ref{section:Optimal Weights in a General System}.


\subsection{Numerical examples}



\begin{example} 
\textbf{Nearly-Block correlation matrix with positive, uniform $\rho$}.\\
In the case of $N=4$, assuming the means and standard deviations are the same and $\alpha_i$'s are 1, {\it i.e.}, $\alpha=[1, 1, 1, 1],\,    \mu= [10, 10, 10, 10],\,\sigma_i \equiv \sigma$ for all $i$. The correlation matrix is $
\left(\begin{array}{llll}
1.  \quad &  0.4\quad &  0\quad & 0\\
  0.4 & 1. &   0.05&  0\\
0 & 0.05 & 1. &     0.4\\
0 &  0 &   0.4&  1.
\end{array}
\right)
$, then for the disjoint group case, there exists one non trivial Nash equilibrium "\{1,3\}-\{2,4\}" . 
\end{example}

\begin{example}
N=10. When we take $\rho_{ij}=0.8$ for all $i\neq j,\,i,j=1,\ldots,N$ except $\rho_{19}=-0.3$. The values for other parameters are listed below:
\begin{equation}
\begin{array}{ll}
&\mu = [1, 1, 1, 2, 2, 3, 6, 6, 6, 7]\\
&\sigma = [4. , 2.8, 1.6, 1. , 3.8, 2.8, 0.9, 1.1, 4.2, 1.8]\\
&\alpha = [\underline{0.4}, 1.2, 1.8, 2.2, \underline{0.4}, 0.9, 2.8, 2.2, \underline{0.4}, 1.9]\\
&B = -8, \quad 
  \text{The initial weights:}\quad (w_{i,1},w_{i,2})=(0.3,0.7)\quad \text{for all }i.
\end{array}
\end{equation}
By the algorithm presented in Section~\ref{sec:algorithm} we find the optimal weights for each individual one by one and it turns out there exists a non-trivial Nash equilibrium in the system:
\begin{equation}
\left(
\begin{array}{cc}
w_{1,1}&w_{1,2}\\
\vdots&\vdots\\
w_{i,1}&w_{i,2}\\
\vdots&\vdots\\
w_{10,1}&w_{10,2}
\end{array}
\right)
=\left(
\begin{array}{cc}
	1.  & 0  \\
       0.51& 0.49\\
       0.48& 0.52\\
       0.44& 0.56\\
       0.  & 1.  \\
       0.49& 0.51\\
       0.44& 0.56\\
       0.45& 0.55\\
       1.  & 0.  \\
       0.49& 0.51
\end{array}
\right).
\end{equation}
We can see for some individuals which seek risks, {\it i.e.}, with extremely small risk aversion parameters, they prefer being alone instead of separated.
\end{example}

\noindent As mentioned at the beginning of the section, the setup of systemic risks in the individual-group structure can be applied to the bank-CCP structure in real-life, where individuals are banks and groups are CCPs. Then an individual with a large utility parameter alpha represents a core bank which is very risk-averse.
\begin{example}\label{example: two central}
N=10 (less risk-averse individuals). In this example, the utility parameters are modified to compare with the previous example and we interpret the results using the ``bank-CCP'' language. Assuming there are two core banks (4,7) and eight peripheral banks, the correlation matrix is given in table \ref{fig:correlation}.

\begin{table}[ht]
\begin{minipage}[h]{0.45\linewidth}
\centering
\includegraphics[width = 0.8\textwidth]{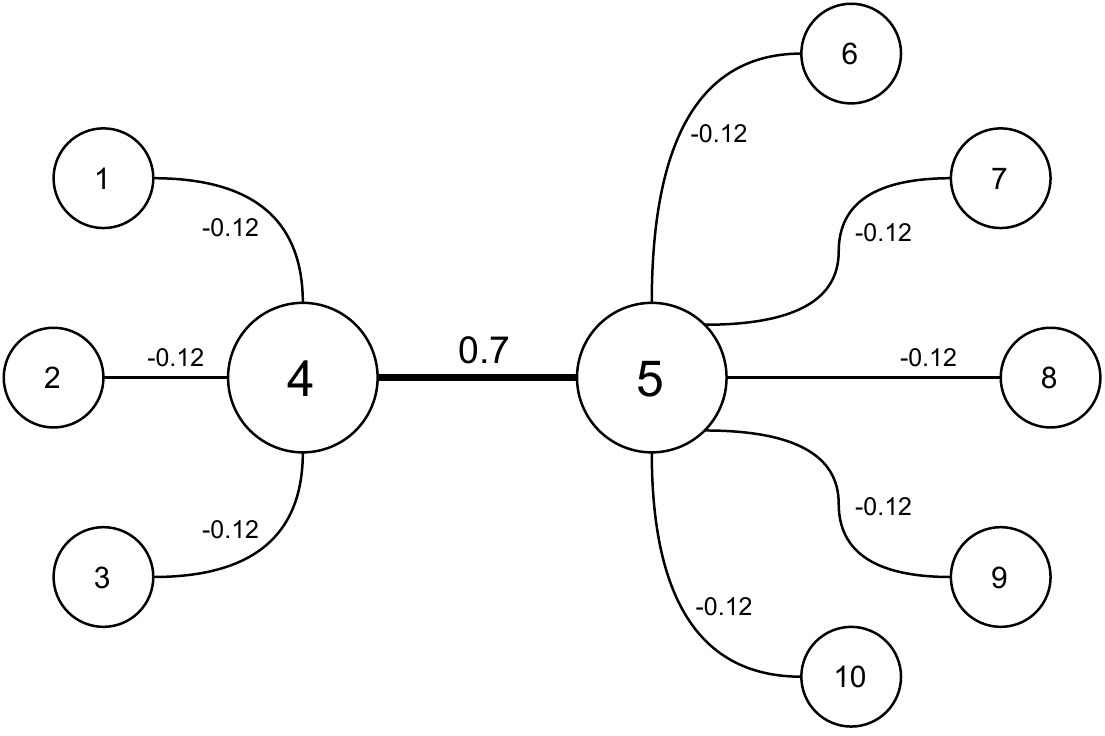}

\end{minipage}\hfill
\begin{minipage}[h]{0.55\linewidth}
\centering
\renewcommand{\arraystretch}{1.5}
\begin{adjustbox}{max width=0.8\textwidth, max height=5cm}
\begin{tabular}{c|ccc:c|c:ccccc}
    &1&2&3&4&5&6&7&8&9&10\\
    \hline
    1&\multicolumn{3}{|c:}{\multirow{3}{*}{$\rho_{pp}=-.25$}}&{}&{}&\multicolumn{5}{c|}{\multirow{3}{*}{$\rho_{pp'}=.05$}}\\
    2&\multicolumn{3}{|c:}{}&\multirow{1}{*}{$\rho_{pc}=$}&\multirow{1}{*}{$\rho_{pc'}=$}&\multicolumn{5}{c|}{}\\
    3&\multicolumn{3}{|c:}{}&{$-.12$}&{$-.09$}&\multicolumn{5}{c|}{}\\
    \hdashline
    4&\multicolumn{3}{|c:}{$\rho_{pc}=-.12$}&&$\rho_{cc'}=.7$&\multicolumn{5}{c|}{$\rho_{cp'}=-.09$}\\
    \hline
    5&\multicolumn{3}{|c:}{$\rho_{pc'}=-.09$}&$\rho_{cc'}=.7$&&\multicolumn{5}{c|}{$\rho_{c'p'}=-.12$}\\
    \hdashline
    6&\multicolumn{3}{|c:}{\multirow{5}{*}{$\rho_{pp'}=.05$}}&{}&{}&\multicolumn{5}{c|}{\multirow{5}{*}{$\rho_{p'p'}=-.25$}}\\
    7&\multicolumn{3}{|c:}{}&{}&{}&\multicolumn{5}{c|}{}\\
    8&\multicolumn{3}{|c:}{}&\multirow{1}{*}{$\rho_{cp'}=$}&\multirow{1}{*}{$\rho_{c'p'}=$}&\multicolumn{5}{c|}{}\\
    9&\multicolumn{3}{|c:}{}&{$-.09$}&{$-.12$}&\multicolumn{5}{c|}{}\\
    10&\multicolumn{3}{|c:}{}&{}&{}&\multicolumn{5}{c|}{}\\
    \hline
\end{tabular}
\end{adjustbox}
   \renewcommand{\arraystretch}{1}
\end{minipage}
\caption{Correlation structure in Example \ref{example: two central}: The left diagram shows partial correlations for 10 banks, where 4 and 5 are core banks and coefficients are labelled for each bank pair. The correlation table on the right shows the correlations matrix where diagonals are all 1 and $\rho_{.,.}$ represents the correlation between two distinct banks. Subscripts $c,c'$ stand for core bank 4 and core bank 5 respectively; $p$ stands for peripheral banks 1,2,3 and $p'$ stands for peripheral banks 6,7,8,9,10.}
\label{fig:correlation}
\end{table}

The values for other parameters are listed below:
\begin{equation*}
\begin{array}{ll}
&\mu = [1, 1, 2, 2, 3, 4, 5, 5, 6, 7]\\
&\sigma = [4. , 2.8, 2.2, 1.7, 1.4, 3.2, 3.8, 1.9, 4.2, 2.5]\\
&\alpha = [0.4, 1. , 1.1, \underline{2.2}, \underline{2.8}, 0.9,0.8, 1.4, 0.6, 1.3]\\
&B = -8,
\end{array}
\end{equation*}
and the initial weights are randomly generated based on uniform distribution between 0 and 1 for every $i$. There exists a non-trivial Nash equilibrium in the system:
\begin{equation}
\left(
\begin{array}{cc}
w_{1,1}&w_{1,2}\\
\vdots&\vdots\\
w_{i,1}&w_{i,2}\\
\vdots&\vdots\\
w_{10,1}&w_{10,2}
\end{array}
\right)
=\left(
\begin{array}{cc}
	1.  & 0.  \\
       1.& 0.\\
       1.& 0.\\
       0.46& 0.54\\
       0.32  & 0.68  \\
       0.  & 1.  \\
       0.& 1.\\
       0.& 1.\\
       0.& 1.\\
       0.  & 1.
\end{array}
\right).
\end{equation}
We can see banks tend to stay with negatively correlated banks to mitigate the systemic risks. And risk-averse banks prefer splitting their risks by joining more CCPs.
\end{example}

\noindent The analysis of the systemic risk and grouping formation can be applied to the reality, and it turns out our numerical results are consistent with the choices of CCPs for banks and financial institutions. One example is shown below using real data.

\begin{example}\label{example: real-life}
\textbf{Real-life Example}. We take two CCPs who are clearing the same products but in different region. One is the Chicago Mercantile Exchange Inc. (CME), which operates two separate clearing services, one for commodity and financial futures and options, and one for interest rate swaps and swaptions. The other one is the European Commodity Clearing (ECC), which is a central clearing house in Europe specialising in energy and commodity products. We select $6$ clearing members and list them in order: J.P. Morgan(JPM), Goldman Sachs(GS), BNP Paribas(BNP), StoneX Group(SNEX), Banco Santander(SAN) and Interactive Brokers Group(IBKR). Among these firms, JPM, GS, SNEX and BNP are members of both CCPs. SAN is only in ECC while IBKR is only in CME.
\begin{table}
\centering
\begin{tabular}{c|cccccc|}
    &JPM&GS&BNP&SNEX&SAN&IBKR\\
    \hline
    &1.&0.82&0.61&0.87&-0.27&0.86\\
    &0.82&1.&0.86&0.83&0.04&0.79\\
    &0.61&0.86&1.&0.65&  0.25& 0.60\\
    &0.87& 0.83&0.65& 1.& -0.35&0.89\\
    &-0.27&  0.04&  0.25& -0.35&  1.&-0.24\\
    &0.86&0.79& 0.60&  0.89&-0.24&1. \\ 
    \hline
    $\sigma$ &0.262& 0.245&0.235&0.264&0.236&0.233\\
    \hline
\end{tabular}
\caption{Correlation matrix and standard deviation for 6 banks in Example \ref{example: real-life}.}
\label{table:real-life example}
\end{table}

We estimate the bank correlation matrix and standard deviation $\sigma$ from banks' stock prices and list them in table \ref{table:real-life example}. Without lost of generality, we assume the expected values of their risks are all $0$ since they have no effect on Nash equilibria according to the formula \eqref{eq: joint-generalsolution}. The values for $B$ is the same as the previous example. The risk-aversion parameters are chosen according clearing members' ``sizes'':
\[
\alpha = [2., 1.8 , 1.7, 1.9, 1.2, 0.85],
\]
and we list them in the following order:  $[\textit{JPM, GS, BNP, SNEX, SAN, IBKR}]$.

The numerical results show that there exists a non-trivial Nash equilibrium 
\begin{equation}
\left(
\begin{array}{cc}
w_{1,1}&w_{1,2}\\
\vdots&\vdots\\
\vdots&\vdots\\
w_{6,1}&w_{6,2}
\end{array}
\right)
=\left(
\begin{array}{cc}
	0.73 & 0.27  \\
       0.61& 0.39\\
       0.56& 0.44\\
       0.54& 0.46\\
       1.  & 0.\\
       0.&1.
\end{array}
\right).
\end{equation}
This is consistent with the fact that the first four firms JPM, GS, BNP, SNEX are spitted and join in both CCP groups while SAN and IBKR belong to different CCPs. However, the distribution of weights cannot be verified here since related data of banks are not revealed in CCP documents.
\end{example}

\section{Conclusion}\label{sec:conclusion}

In this paper, we generalize the systemic risk measure proposed in \cite{biagini2019unified, biagini2020fairness} by allowing individual banks to choose their preferred groups instead of being assigned to certain groups. This introduces realistic game features in the proposed models, and allows us to analyze the systemic risk for disjoint and overlapping groups ({\it e.g.}, central clearing counterparties (CCP)). We introduce the concept of Nash equilibrium for these new models, and analyze the optimal solution under the Gaussian distribution of the risk factor. We also provide an explicit solution for the individual banks' risk allocation and study the existence and uniqueness of Nash equilibrium both theoretically and numerically. The developed numerical algorithm can simulate scenarios of equilibrium, and we apply it to study the bank-CCP structure with real data and show the validity of the proposed model. Further research includes obtaining more actual data on bank balances and bank interconnections to conduct more in-depth research and analysis. The participation percentage of financial institutions is left to be validated and explained with more data. It is also valuable to consider CCP clearing fee charge as in \cite{Hu2018SysCCP} and its effect on the equilibrium.

\section*{Acknowledgement}

The authors are grateful to St\'{e}phane Cr\'{e}pey, Samuel Drapeau, Mekonnen Tadese, Dorinel Bastide and Romain Arribehaute  for useful discussions on the formation of CCPs.
J.-P. Fouque acknowledges the support by the NSF grant DMS-1814091. T. Ichiba was supported in part by the NSF grant DMS-2008427. R. Hu was partially supported by the NSF grant DMS-1953035, the Faculty Career Development Award, the Research Assistance Program Award, and the
Early Career Faculty Acceleration funding at University of California, Santa Barbara.



\bibliographystyle{acm.bst}
\bibliography{references}

\appendix
\section{Appendix}\label{Appendix}

\subsection{Comparison between Trivial Grouping and Multi-Groups}\label{proof: compare-Trivial-Multi-Groups}
We first look at the trivial grouping, {\it i.e.}, $m=h=1$ and all $k\in I_1=\{1,2,\ldots,N\}$. The group parameter and group vectors are
$$
\beta_m=\beta=\sum\limits_{i=1}^N\frac{1}{\alpha_i},\quad A_m=(1,1,\ldots,1).
$$
Then following \eqref{eq: generalsolution}, the systemic risk allocation of individual $i$ is:
\[
\mathbb E _{\mathbb Q^{1}_{\mathbf X}} [ Y^{i}_{\mathbf X} ]
=-\mu_i+\frac{1}{\alpha_i}\log(\frac{\beta}{-B})+\frac{1}{\beta}\sum\limits_{j=1}^N\sigma_{ji} -\frac{1}{2\beta^2\alpha_i} \sum\limits_{k,j=1}^N \sigma_{jk}.
\]
The total systemic risk allocation for the system is
\begin{align}
\sum\limits_{i=1}^N \mathbb E _{\mathbb Q^{1}_{\mathbf X}} [ Y^{i}_{\mathbf X} ]
&=-\sum\limits_{i=1}^N\mu_i+\beta\,\log(\frac{\beta}{-B})+\frac{1}{2\beta}  \sum\limits_{i,j=1}^N \sigma_{ji}.\label{eq: totalrisk-trivial}
\end{align}

Then for the multi-group case, assuming $m=1,\ldots, h$ ($h\geq 2$) and for $k_m\in I_m\neq \emptyset$, the systemic risk allocation of individual $k_m$ is:
\begin{align}
\mathbb E _{\mathbb Q^{m}_{\mathbf X}} [ Y^{k_m}_{\mathbf X} ]
&=-\mu_{k_m}+\frac{1}{\alpha_{k_m}}\log(\frac{\beta}{-B})+\frac{1}{\beta_{m}}\sum\limits_{j\in I_m}\sigma_{j{k_m}} -\frac{1}{2\beta_m^2\alpha_{k_m}} \sum\limits_{ j,l\in I_m} \sigma_{jl}.\nonumber
\end{align}
The total risk allocation is:
\begin{align}
\sum\limits_{m=1}^h\sum\limits_{k_m\in I_m}\mathbb E _{\mathbb Q^{m}_{\mathbf X}} [ Y^{k_m}_{\mathbf X} ]
&=-\sum\limits_{i=1}^N\mu_{i}+\beta\,\log(\frac{\beta}{-B}) +\frac{1}{2}\sum\limits_{m=1}^h\frac{1}{\beta_m}\sum\limits_{j,k\in I_m}\sigma_{j{k}}.\label{eq: totalrisk-nontrivial}
\end{align}
We need to compare the total risk of trivial grouping \eqref{eq: totalrisk-trivial} and the total risk of nontrivial grouping \eqref{eq: totalrisk-nontrivial}. For simplicity, we take $h=2$ in the nontrivial grouping case and compare.

Assume a $N$-individual system is divided into two subgroups with sizes $N_1=|I_1|$, $N_2=|I_2|$, respectively. Given all risk factors define
\[
S=\sum\limits_{i=1}^N X^i,\quad 
S_1=\sum\limits_{i\in I_1} X^i,\quad 
S_2=\sum\limits_{i\in I_2} X^i.
\]
Note that $S=S_1+S_2$. Therefore
\[
\Var(S)=\sum\limits_{i,j=1}^N \sigma_{ij},\quad
\Var(S_1)=\sum\limits_{i,j\in I_1} \sigma_{ij},\quad
\Var(S_2)=\sum\limits_{i,j\in I_2} \sigma_{ij}.
\]
Then we compare the last term in \eqref{eq: totalrisk-trivial}, \eqref{eq: totalrisk-nontrivial}:
\begin{align}
&\eqref{eq: totalrisk-trivial}: \quad\dfrac{1}{2\beta}\sum\limits_{i,j=1}^N\sigma_{i,j}=\dfrac{1}{2\beta} \Var(S)\label{eqn: lastterm-totalrisk-trivial}
\\
&\eqref{eq: totalrisk-nontrivial}:\quad\dfrac{1}{2}\left(\dfrac{1}{\beta_1}\,\sum\limits_{j,k\in I_1}\sigma_{j,k}+\dfrac{1}{\beta_2}\,\sum\limits_{j,k\in I_2}\sigma_{j,k}\right) =\dfrac{1}{2}\left(\dfrac{1}{\beta_1}\,\Var(S_1)+\dfrac{1}{\beta_2}\,\Var(S_2)\right).\label{eqn: lastterm-totalrisk-nontrivial}
\end{align}
Since 
\begin{align*}
\Var(S)=\Var(S_1)+\Var(S_2)+2\Cov (S_1,S_2),
\end{align*}
\begin{align*}
2\beta_1\beta_2\beta\cdot \left(\ref{eqn: lastterm-totalrisk-nontrivial}-\ref{eqn: lastterm-totalrisk-trivial}\right)
&=\beta_1\beta_2\Var(S_1)+\beta_2^2\Var(S_1)+\beta_1^2\Var(S_2)+\beta_1\beta_2\Var(S_2)
\\
&\quad -\beta_1\beta_2\left(\Var(S_1)+\Var(S_2)+2\Cov (S_1,S_2)\right)
\\
&=\beta_2^2\Var(S_1)+\beta_1^2\Var(S_2)-2\beta_1\beta_2\Cov (S_1,S_2)
\\
&=\Var(\beta_2 S_1-\beta_1 S_2)\geq 0.
\end{align*}
Thus \eqref{eqn: lastterm-totalrisk-trivial} $\leq$  \eqref{eqn: lastterm-totalrisk-nontrivial}, which is equivalent to \eqref{eq: totalrisk-trivial} $\leq$ \eqref{eq: totalrisk-nontrivial}. (The equality holds only if $S_1= S_2=S$, which we can exclude.) 

We can conclude the trivial grouping has a smaller total systemic risk allocation for the $N$-player system compared with two-group case. It can be extended to general grouping case and show the advantage of trivial grouping or fewer groups in terms of the total risk. And it is consistent with the monotonicity property proved in \cite{biagini2020fairness}.

\subsection{Proof of Claim \ref{claim: 3.1} }\label{proof: claim 3.1}
Under the assumption, $\beta_m=|I_m|\,\frac{1}{\alpha}$, equation \eqref{eq: generalsolution} becomes:
\begin{align*}
\mathbb E _{\mathbb Q^{m}_{\mathbf X}} [ Y^{k}_{\mathbf X} ] 
&= -\mu_k+\frac{1}{\alpha}\log(\frac{\beta}{-B})+\frac{\alpha}{|I_m|}(\sigma^2+(|I_m|-1)\rho\sigma^2)\\
& -\frac{1}{2\alpha}(\frac{\alpha}{|I_m|})^2 \big( |I_m|\sigma^2+\binom{|I_m|}{2} 2\rho\sigma^2\big)\\
&=-\mu_k+\frac{1}{\alpha}\log(\frac{\beta}{-B})+\frac{\alpha}{2|I_m|}\sigma^2+\big(\alpha\frac{|I_m|-1}{|I_m|}-\alpha\frac{|I_m|-1}{2|I_m|}\big)\rho\sigma^2\\
&=-\mu_k+\frac{1}{\alpha}\log(\frac{\beta}{-B})+\alpha\frac{1}{2|I_m|}\sigma^2+\alpha\frac{|I_m|-1}{2|I_m|}\rho\sigma^2\\
&=-\mu_k+\frac{1}{\alpha}\log(\frac{\beta}{-B})+
\frac{\alpha}{2}\big((1-\rho)\frac{1}{|I_m|}+\rho\big)\sigma^2.
\end{align*}
First, when $\rho\neq 1$, the function is monotonically decreasing in $|I_m|$. So the risk allocation for every individual/individual is maximized when $|I_m|=n$, {\it i.e.}, all individuals/individuals are in the same group.

If we assume individuals are separated in several groups which makes it a Nash for all. For some individual $i$ in the second largest group, moving to the largest group makes it achieve a smaller risk allocation which is better. If there are only two equal size groups, {\it i.e.}, "$n/2-n/2$", one individual $i$ in group 1 joining the other one gives "$(n/2-1)-(n/2+1)$" and $\mathbb E _{\mathbb Q^{1}_{\mathbf X}} [ Y^{i}_{\mathbf X} ]>\mathbb E _{\mathbb Q^{2}_{\mathbf X}} [ Y^{i}_{\mathbf X} ]  $. So "$(n/2-1)-(n/2+1)$" cannot be Nash neither. 

In conclusion, nontrivial grouping strategy cannot be Nash and only $|I_m|=n$ is Nash under the case that all standard deviation and utility parameters are the same, and the correlation coefficient $\rho\in [-1,1)$. When $\rho = 1$, risks for every individual is constant and grouping has no effect.

\subsection{Proof of Claim \ref{claim: const_block_rho}}\label{proof: claim 3.2}
First we look at the grouping "$\{1,2\}-\{3,4\}$" and find the equivalent condition to make it a Nash equilibrium. 
We compare risks for individual 1 under two cases, according to \eqref{eq: generalsolution},
\begin{align}
\quad &\text{under } "\{1,2\}-\{3,4\}":\quad \mathbb{E}_{{\mathbb Q}_{\mathbf X}^1}\,[Y_{\mathbf X}^1]=-\mu_1+\log(\frac{\beta}{-B})+\frac{1}{2}(\sigma^2+\rho\sigma^2)-\frac{1}{8}(2\sigma^2+2\rho\sigma^2);\nonumber\\
&\text{under } "\{2\}-\{1,3,4\}":
\quad \mathbb{E}_{{\mathbb Q}_{\mathbf X}^2}\,[Y_{\mathbf X}^1]=-\mu_1+\log(\frac{\beta}{-B})+\frac{1}{3}\sigma^2-\frac{1}{18}(3\sigma^2+2\rho\sigma^2).\nonumber
\end{align}
For individual 1, to make the risk allocation under "$\{1,2\}-\{3,4\}$" $\leq$ that under "$\{2\}-\{1,3,4\}$":
\begin{align}
&\frac{1}{3}\sigma^2-\frac{1}{18}(3\sigma^2+2\rho\sigma^2)\geq\frac{1}{2}(\sigma^2+\rho\sigma^2)-\frac{1}{8}(2\sigma^2+2\rho\sigma^2)\quad \Leftrightarrow\quad \rho\leq -\frac{3}{13}. \nonumber
\end{align}
For individuals 2,3,4, we repeat similar discussion and the condition is the same: $\rho\leq -\frac{3}{13}$.

So in conclusion, if $\rho\leq -\frac{3}{13}$, grouping "$\{1,2\}-\{3,4\}$" is a Nash equilibrium.

\noindent  In the grouping "$\{1,3\}-\{2,4\}$", we follow the same discussion about comparing risk allocations for all individuals. For example, for individual 1:
\begin{align}
\quad &\text{under } "\{1,3\}-\{2,4\}":\quad \mathbb{E}_{{\mathbb Q}_{\mathbf X}^1}\,[Y_{\mathbf X}^1]=-\mu_1+\log(\frac{\beta}{-B})+\frac{1}{2}\sigma^2-\frac{1}{8}2\sigma^2;\nonumber\\
&\text{under } "\{3\}-\{1,2,4\}":
\quad \mathbb{E}_{{\mathbb Q}_{\mathbf X}^2}\,[Y_{\mathbf X}^1]=-\mu_1+\log(\frac{\beta}{-B})+\frac{1}{3}(\sigma^2+\rho\sigma^2)-\frac{1}{18}(3\sigma^2+2\rho\sigma^2).\nonumber
\end{align}
For individual 1, to make the risk allocation under "$\{1,3\}-\{2,4\}$" $\leq $ that under "$\{3\}-\{1,2,4\}$":
\begin{align}
&\frac{1}{3}(\sigma^2+\rho\sigma^2)-\frac{1}{18}(3\sigma^2+2\rho\sigma^2)\geq\frac{1}{2}\sigma^2-\frac{1}{8}2\sigma^2\quad \Leftrightarrow\quad \rho\geq \frac{3}{8}. \nonumber
\end{align}
For individuals 2,3,4, we repeat similar discussion and the condition is the same: $\rho\geq \frac{3}{8}$.

As a result, grouping "$\{1,3\}-\{2,4\}$" is a Nash equilibrium if $\rho\geq \frac{3}{8}$. It is also the equivalent condition for grouping "$\{1,4\}-\{2,3\}$" to be a Nash equilibrium

\subsection{Proof of Claim \ref{clami:const_block_rhofor5}}\label{proof: claim 3.3}
Similar to the proof of claim \ref{claim: const_block_rho}, we look at the grouping "$\{1,2\}-\{3,4,5\}$" and find the equivalent condition to make it a Nash equilibrium. Under the assumptions, we first compare risks for individual 1 under two cases, according to \eqref{eq: generalsolution},
\begin{align}
\quad &\text{under } "\{1,2\}-\{3,4,5\}":\quad \mathbb{E}_{{\mathbb Q}_{\mathbf X}^1}\,[Y_{\mathbf X}^1]=-\mu_1+\log(\frac{\beta}{-B})+\frac{1}{2}(\sigma^2+\rho\sigma^2)-\frac{1}{8}(2\sigma^2+2\rho\sigma^2);\nonumber\\
&\text{under } "\{2\}-\{1,3,4,5\}":
\quad \mathbb{E}_{{\mathbb Q}_{\mathbf X}^2}\,[Y_{\mathbf X}^1]=-\mu_1+\log(\frac{\beta}{-B})+\frac{1}{4}\sigma^2-\frac{1}{32}(4\sigma^2+6\rho\sigma^2).\nonumber
\end{align}
For individual 1, to make the risk allocation under "$\{1,2\}-\{3,4,5\}$" $\leq$ that under "$\{2\}-\{1,3,4,5\}$":
\begin{align}
&\frac{1}{4}\sigma^2-\frac{1}{32}(4\sigma^2+6\rho\sigma^2)\geq\frac{1}{2}(\sigma^2+\rho\sigma^2)-\frac{1}{8}(2\sigma^2+2\rho\sigma^2)\quad \Leftrightarrow\quad \rho\leq -\frac{2}{7}. \nonumber
\end{align}
For individual 2, this is the same condition to have a smaller risk allocation. We then compare risks for individual 3 under two cases:
\begin{align}
\quad &\text{under } "\{1,2\}-\{3,4,5\}":\quad \mathbb{E}_{{\mathbb Q}_{\mathbf X}^2}\,[Y_{\mathbf X}^3]=-\mu_3+\log(\frac{\beta}{-B})+\frac{1}{3}(\sigma^2+2\rho\sigma^2)-\frac{1}{18}(3\sigma^2+6\rho\sigma^2);\nonumber\\
&\text{under } "\{1,2,3\}-\{4,5\}":
\quad \mathbb{E}_{{\mathbb Q}_{\mathbf X}^1}\,[Y_{\mathbf X}^3]=-\mu_3+\log(\frac{\beta}{-B})+\frac{1}{3}\sigma^2-\frac{1}{18}(3\sigma^2+2\rho\sigma^2).\nonumber
\end{align}
To make the risk allocation for individual 3 under "$\{1,2\}-\{3,4,5\}$" $\leq$
that under "$\{1,2,3\}-\{4,5\}$", we get the equivalent condition $\rho\leq 0$. This is true for individuals 4,5 as well.

In conclusion, when $\rho\leq -2/7$, "$\{1,2\}-\{3,4,5\}$" is a Nash equilibrium for all individuals.

In the grouping "$\{1,3\}-\{2,4,5\}$", we follow the same discussion about comparing risk allocations for all individuals. Individual 1 will stay with individual 3 instead of joining the other group if $\rho = 1$. Individual 2 will stay if $\rho \geq 0$. Individual 3 will stay with individual 1 if $\rho\geq \frac{2}{5}$ and individuals 4,5 won’t move for any $\rho$. As a result, grouping "$\{1,3\}-\{2,4,5\}$" is not a Nash equilibrium for any $\rho\in (-1,1)$.

For the grouping "$\{1,2,3\}-\{4,5\}$", after discussion about the condition for every individual not moving, we find a contradiction which proves for any value of $\rho$, this grouping cannot be a Nash.

\subsection{Proof of Theorem \ref{thm:JointResult}}\label{appendix:JointResult}

We can rewrite the systemic risk measure $\rho$ as: 
\begin{align}
{\bm \rho}(\mathbf{X})=\inf\Big\{\sum\limits_{m=1}^h d_m\,:\,
&\, d=(d_1,\ldots , d_h)\in \mathbb{R}^h, \mathbf{Y}=(Y^{i,j},i\in I_j,\,1\leq j\leq h)\in L^0(\mathbb{R}^{\sum_{m=1}^h |I_m|* h}),\nonumber
\\
&\mathbb{E}\left[\sum\limits_{m=1}^h\,\sum\limits_{k\in I_m}-\frac{1}{\alpha_k}\exp[-\alpha_k(w_{k,m}X^k+Y^{k,m})]\right]\,=B\,,\nonumber
\\
&
\sum\limits_{i\in I_j} Y^{i,j}=d_j, \,\text{for } j=1,2,\ldots, h
\Big\},
\end{align}
where $(I_1,\ldots, I_h)$ are the group index sets given. {For any group $m$, we define the smallest element as $m^0\in I_m$ and fix another element $m^*\in I_m$ and $m^*\neq m^0$.} (When there is only one element in the group, $m^*=m^0$ and the discussion will be similar.) In the following proof we assume $|I_m|\geq 2$ for all $m=1,\ldots, h$. 
We also assume $Y$ is defined on a finite space: $Y^{k,m}\in \{y^{k,m}_1,\ldots, y^{k,m}_M\}$ for all $k,m$ and $y^{k,m}_j\in\mathbb{R}$. Then we have: $y^{m^*,m}_{j}=d_m-\sum\limits_{k\in I_m,k\neq m^*} y^{k,m}_j$ for $m=1,\ldots, h$.

We show the results with the Lagrange's method with the function defined by 

\begin{align}
\mathcal{L}\,(\mathbf{d},\mathbf{Y},\lambda)=\, 
\sum_{m=1}^{h} d_m \, +\lambda\Bigg\{&\sum\limits_{j=1}^M p_j\,\sum\limits_{m=1}^h\,\left[\sum\limits_{k\in I_m,k\neq m^*}\,\frac{-1}{\alpha_k}\exp\left(-\alpha_k(w_{k,m}X^k(\omega_j)+y^{k,m}_j)\right)\right.\nonumber
\\
&
\left. +\frac{-1}{\alpha_{m^*}}\exp\left(-\alpha_{m^*}\bigg(w_{{m^*},m}X^{m^*}(\omega_j)+d_m-\sum\limits_{k\in I_m,k\neq m^*}y^{k,m}_j\bigg)\right)\right]
-B\,\Bigg\}.
\end{align}
We compute partial derivatives of $\mathcal{L}$ with respect to all variables and get all equivalent condition in the following.
\begin{enumerate}
\item Given $m$, for $j=1,\ldots, M$, $k\in I_m$ and $k\neq m^*$: $\frac{\partial \mathcal{L}}{\partial y^{k,m}_j}=0$\\
if and only if for every fixed $j$,
\begin{align}
&0=\frac{\partial \mathcal{L}}{\partial y^{k,m}_j}
=\lambda p_j\left[ e^{-\alpha_k\left(w_{k,m}X^k(\omega_j)+y^{k,m}_j\right)} -e^{-\alpha_{m^*}\left(w_{m^*,m}X^{m^*}(\omega_j)+d_m-\sum\limits_{k\in I_m,k\neq m^*}y^{k,m}_j\right)} \right]\nonumber
\\
&\Longleftrightarrow\quad  \exp \left(-\alpha_k\left(w_{k,m}X^k(\omega_j)+y^{k,m}_j\right)\right) =\exp \left(-\alpha_{m^*}\left(w_{m^*,m}X^{m^*}(\omega_j)+d_m-\sum\limits_{k\in I_m,k\neq m^*}y^{k,m}_j\right)\right) \label{eqn: thmproof-tool1}
\\
&\text{It implies}\quad
 \exp\left(-\alpha_k\left(w_{k,m}X^k(\omega_j)+y^{k,m}_j\right)\right) =\exp \left(-\alpha_{m^0}\left(w_{m^0,m}X^{m^0}(\omega_j)+y^{m^0,m}_j\right)\right) .\nonumber
\end{align}
Then for $k\neq m^*$,
\begin{equation}\label{eqn: thmproof-tool2}
y^{k,m}_j=\frac{1}{\alpha_k}\left(\alpha_{m^0}w_{{m^0},m}X^{m^0}(\omega_j)-\alpha_{k}w_{{k},m}X^{k}(\omega_j)+\alpha_{m^0}y^{{m^0},m}_j\right),
\end{equation}
and by \eqref{eqn: thmproof-tool1} and \eqref{eqn: thmproof-tool2}, we obtain that (details are shown below):
\begin{align}
y^{m^0,m}_j&=\frac{1}{\alpha_{m^0}\beta_m}\left(\sum\limits_{k\in I_m}w_{{k},m}X^{k}(\omega_j) \right)-w_{{m^0},m}X^{m^0}(\omega_j)+\frac{1}{\alpha_{m^0}\beta_m}d_m\nonumber
\\
&=\frac{1}{\alpha_{m^0}\beta_m}S_m(\omega_j)-w_{{m^0},m}X^{m^0}(\omega_j)+\frac{1}{\alpha_{m^0}\beta_m}d_m,
\label{eqn: thmproof-tool3}
\end{align}
where $\beta_m=\sum\limits_{k\in I_m}\frac{1}{\alpha_k}$, $S_m=\sum\limits_{k\in I_m}w_{k,m}X^k$.

\begin{proof}
By \eqref{eqn: thmproof-tool2}, 
\begin{align*}
\sum\limits_{k\in I_m,k\neq m^*}y^{k,m}_j&=\sum\limits_{k\in I_m,k\neq m^*}\frac{1}{\alpha_k}\alpha_{m^0}\left(w_{{m^0},m}X^{m^0}(\omega_j)+y^{{m^0},m}_j\right)-\sum\limits_{k\in I_m,k\neq m^*}w_{{k},m}X^{k}(\omega_j)
\\
&=(\beta_m-\frac{1}{\alpha_{m^*}})\alpha_{m^0}\left(w_{{m^0},m}X^{m^0}(\omega_j)+y^{{m^0},m}_j\right)-\sum\limits_{k\in I_m,k\neq m^*}w_{{k},m}X^{k}(\omega_j),
\end{align*}
\begin{equation*}
w_{m^*,m}X^{m^*}(\omega_j)+d_m-\sum\limits_{k\in I_m,k\neq m^*}y^{k,m}_j=d_m-(\beta_m-\frac{1}{\alpha_{m^*}})\alpha_{m^0}\left(w_{{m^0},m}X^{m^0}(\omega_j)+y^{{m^0},m}_j\right)+S_m(\omega_j).
\end{equation*}
Thus using \eqref{eqn: thmproof-tool1},
\begin{align*}
&-\alpha_{m^*}\left(w_{m^*,m}X^{m^*}(\omega_j)+d_m-\sum\limits_{k\in I_m,k\neq m^*}y^{k,m}_j\right)
\\
&\quad\quad\quad\quad\quad\quad =-\alpha_{m^*}d_m+(\beta_m\alpha_{m^*}-1)\alpha_{m^0}\left(w_{{m^0},m}X^{m^0}(\omega_j)+y^{{m^0},m}_j\right)-\alpha_{m^*}S_m(\omega_j)
\\
&\text{by }\eqref{eqn: thmproof-tool1}\quad\quad =-\alpha_{m^0}\left(w_{m^0,m}X^{m^0}(\omega_j)+y^{m^0,m}_j\right)
\\
&\Longrightarrow\quad
\alpha_{m^*}\beta_m\alpha_{m^0}\left(w_{{m^0},m}X^{m^0}(\omega_j)+y^{{m^0},m}_j\right)=\alpha_{m^*}d_m+\alpha_{m^*}S_m(\omega_j)
\\
&\Longrightarrow\quad
\eqref{eqn: thmproof-tool3}.
\end{align*}
\end{proof}

\item For $m=1,\ldots, h$, the derivative with respect to $d_m$:
$\frac{\partial \mathcal{L}}{\partial d_m}=0$ if and only if 
\begin{equation}\label{eqn: thmproof-tool4}
0=\frac{\partial \mathcal{L}}{\partial d_m}=1+\lambda\sum\limits_{j=1}^Mp_j\, \exp\left(-\alpha_{m^*}\bigg(w_{{m^*},m}X^{m^*}(\omega_j)+d_m-\sum\limits_{k\in I_m,k\neq m^*}y^{k,m}_j\bigg)\right)
\end{equation}
By \eqref{eqn: thmproof-tool1},
\begin{align}
0&=1+\lambda\sum\limits_{j=1}^M p_j\, \exp\left(-\alpha_{k}\big(w_{{k},m}X^{k}(\omega_j)+y^{k,m}_j\big)\right) \quad\text{for all }k\in I_m,\,\text{and } k\neq m^*,\nonumber
\\
\text{{\it i.e.},}\quad&\sum\limits_{j=1}^M p_j\, \exp\left(-\alpha_{k}\big(w_{{k},m}X^{k}(\omega_j)+y^{k,m}_j\big)\right)=\dfrac{-1}{\lambda}.\label{eqn: thmproof-tool5}
\end{align}

\item The derivative with respect to $\lambda$:
$\frac{\partial \mathcal{L}}{\partial \lambda}=0$ if and only if 
\begin{align*}
B&=\mathbb{E}\left[\sum\limits_{m=1}^h\,\sum\limits_{k\in I_m}\frac{-1}{\alpha_k}\exp[-\alpha_k(w_{k,m}X^k+Y^{k,m})]\right]
\\
&=\sum\limits_{j=1}^M p_j\cdot\sum\limits_{m=1}^h\,\left[\sum\limits_{k\in I_m,k\neq m^*}\,\frac{-1}{\alpha_k}\exp\left(-\alpha_k(w_{k,m}X^k(\omega_j)+y^{k,m}_j)\right)\right.\nonumber
\\
&
\quad\quad\quad\quad\quad\quad\quad\left. +\frac{-1}{\alpha_{m^*}}\exp\left(-\alpha_{m^*}\bigg(w_{{m^*},m}X^{m^*}(\omega_j)+d_m-\sum\limits_{k\in I_m,k\neq m^*}y^{k,m}_j\bigg)\right)\right]
\\
&=\sum\limits_{j=1}^M \sum\limits_{m=1}^h\,\left[\sum\limits_{k\in I_m,k\neq m^*}\,\frac{-1}{\alpha_k}\cdot p_j\,\exp\left(-\alpha_k(w_{k,m}X^k(\omega_j)+y^{k,m}_j)\right)\right.\nonumber
\\
&
\quad\quad\quad\quad\quad\quad\quad\left. +\frac{-1}{\alpha_{m^*}}\cdot p_j\,\exp\left(-\alpha_{m^*}\bigg(w_{{m^*},m}X^{m^*}(\omega_j)+d_m-\sum\limits_{k\in I_m,k\neq m^*}y^{k,m}_j\bigg)\right)\right]
\\
&\text{by }\eqref{eqn: thmproof-tool4},\,\eqref{eqn: thmproof-tool5}\quad =\sum\limits_{m=1}^h\left[ \sum\limits_{k\in I_m,k\neq m^*}\,\frac{-1}{\alpha_k}\,\frac{-1}{\lambda}+\frac{-1}{-\alpha_{m^*}}\,\frac{-1}{\lambda}\right]
\\
&=\frac{1}{\lambda}\sum\limits_{m=1}^h \sum\limits_{k\in I_m}\,\frac{1}{\alpha_k}=\frac{1}{\lambda}\,\beta.
\end{align*}
Hence, 
\begin{equation}\label{eqn:eqn: thmproof-tool6}
\lambda=\dfrac{\beta}{B}.
\end{equation}
\end{enumerate}

\noindent $\bullet$ We then compute $d_m$ by inserting \eqref{eqn:eqn: thmproof-tool6} and \eqref{eqn: thmproof-tool3} in \eqref{eqn: thmproof-tool5} for $k=m^0\in I_m$, and $m^0\neq m^*$:
\begin{align*}
-\dfrac{B}{\beta}
&=\sum\limits_{j=1}^M p_j\, \exp\left(-\alpha_{m^0}\big(w_{{m^0},m}X^{m^0}(\omega_j)+y^{{m^0},m}_j\big)\right)
\\
&=\sum\limits_{j=1}^M p_j\, \exp\left(-\frac{1}{\beta_m}\big(S_m(\omega_j)+d_m\big)\right)
\\
&=e^{-\frac{d_m}{\beta_m}}\sum\limits_{j=1}^M p_j\, \exp\left(-\frac{1}{\beta_m}S_m(\omega_j)\right)
=e^{-\frac{d_m}{\beta_m}}\,\mathbb{E}\left(e^{-\frac{S_m}{\beta_m}}\right).
\end{align*}
So 
\begin{equation}
d_m=\beta_m\log\left(\dfrac{\beta}{-B}\,\mathbb{E}\left(e^{-\frac{S_m}{d_m}}\right)  \right).
\end{equation}
Then back to \eqref{eqn: thmproof-tool2} and \eqref{eqn: thmproof-tool3}, for $k,m^0, m^*\in I_m$ and $k\neq m^0\neq m^*$:
\begin{align*}
y^{m^0,m}_j
&=\frac{1}{\alpha_{m^0}\beta_m}S_m(\omega_j)-w_{{m^0},m}X^{m^0}(\omega_j)+\frac{1}{\alpha_{m^0}\beta_m}d_m;
\\
y^{k,m}_j&=\frac{\alpha_{m^0}}{\alpha_k}w_{{m^0},m}X^{m^0}(\omega_j)-w_{{k},m}X^{k}(\omega_j)
\\
&\quad\quad +\frac{1}{\alpha_k}\left[\frac{1}{\beta_m}S_m(\omega_j)-\alpha_{m^0} w_{{m^0},m}X^{m^0}(\omega_j)+\frac{1}{\beta_m}d_m \right]
\\
&=-w_{{k},m}X^{k}(\omega_j)+\frac{1}{\alpha_k\beta_m}\left(S_m(\omega_j)+d_m \right);
\\
y^{m^*,m}_j&=d_m-\sum\limits_{k\in I_m,k\neq m^*}y^{k,m}_j
\\
&=d_m+\sum\limits_{k\in I_m,k\neq m^*} w_{{k},m}X^{k}(\omega_j)-\sum\limits_{k\in I_m,k\neq m^*} \frac{1}{\alpha_k\beta_m}\left(S_m(\omega_j)+d_m \right)
\\
&=d_m+\left(S_m(\omega_j)-w_{{m^*},m}X^{m^*}(\omega_j)\right)-\left(\beta_m-\frac{1}{\alpha_{m^*}}\right)\frac{1}{\beta_m}\left(S_m(\omega_j)+d_m\right)
\\
&=-w_{{m^*},m}X^{m^*}(\omega_j)+\frac{1}{\alpha_{m^*}\beta_m}\left(S_m(\omega_j)+d_m \right).
\end{align*}
So for given $m$, for all $k\in I_m$, we have:
\[
Y^{k,m}=-w_{{k},m}X^{k}+\frac{1}{\alpha_{k}\beta_m}\left(S_m+d_m \right)
\]
where $
d_m=\beta_m\log\left(\dfrac{\beta}{-B}\,\mathbb{E}\left(e^{-\frac{S_m}{d_m}}\right)  \right).
$

In addition, the systemic risk measure is given by:
\begin{align*}
{\bm \rho} (\mathbf{X})&=\sum\limits_{m=1}^h d_m=\sum\limits_{m=1}^h\beta_m\log\left(\dfrac{\beta}{-B}\,\mathbb{E}\left(e^{-\frac{S_m}{d_m}}\right)  \right)
\\
&=\beta\log\left(\dfrac{\beta}{-B}\right)+\sum\limits_{m=1}^h\beta_m\log\left(\mathbb{E}\left(e^{-\frac{S_m}{d_m}}\right)  \right).
\end{align*}

\subsection{Proof of Proposition \ref{proposition: sensitivity analysis-main}}
\label{Appendix: sensitivity analysis-main}
First, we prove the marginal risk allocation for individual $i$ in group $j$. \\
By theorem \ref{thm:JointResult}, for $i\in I_j$
\begin{align*}
\mathbb{E}_{\mathbb{Q}^j_{\mathbf{X}+\varepsilon\mathbf{Z}}}\left[Y^{i,j}_{\mathbf{X}+\varepsilon\mathbf{Z}}\right]
&=\mathbb{E}\,\left[Y^{i,j}_{\mathbf{X}+\varepsilon\mathbf{Z}}\cdot \dfrac{\mathrm{d}\mathbb{Q}^j_{\mathbf{X}+\varepsilon\mathbf{Z}}}{\mathrm{d}\mathbb{P}}\right]
\\
&=\mathbb{E}\,\left[\left(-w_{i,j}(X^{i}+\varepsilon Z^i)+\frac{1}{\alpha_i\beta_j}\left(S_j+\varepsilon S_j^\mathbf{Z}\right)+\frac{1}{\alpha_i\beta_j}d_j^{\mathbf{X}+\varepsilon \mathbf{Z}}\right)\cdot\dfrac{\exp\left({-\frac{S_j^{\mathbf{X}+\varepsilon \mathbf{Z}}}{\beta_j}}\right)}{\mathbb{E}\left[\exp\left({-\frac{S_j^{\mathbf{X}+\varepsilon \mathbf{Z}}}{\beta_j}}\right]\right)}\right]
\\
&=-w_{i,j}\,\dfrac{\mathbb{E}\left((X^i+\varepsilon Z^i)\exp\left({-\frac{S_j^{\mathbf{X}+\varepsilon \mathbf{Z}}}{\beta_j}}\right)\right)}{\mathbb{E}\left(\exp\left({-\frac{S_j^{\mathbf{X}+\varepsilon \mathbf{Z}}}{\beta_j}}\right)\right)}
+\frac{1}{\alpha_i\beta_j}\,\dfrac{\mathbb{E}\left((S_j+\varepsilon S_j^\mathbf{Z})\exp\left({-\frac{S_j^{\mathbf{X}+\varepsilon \mathbf{Z}}}{\beta_j}}\right)\right)}{\mathbb{E}\left(\exp\left({-\frac{S_j^{\mathbf{X}+\varepsilon \mathbf{Z}}}{\beta_j}}\right)\right)}
\\
&+\frac{1}{\alpha_i}\left(\log (-\beta/B)+\log \left(\mathbb{E}\left(e^{-\frac{S_j^{\mathbf{X}+\varepsilon\mathbf{Z}}}{\beta_j}}\right)\right)\right)
\\
&=\RomanNumeralCaps{1}+\RomanNumeralCaps{2}+\RomanNumeralCaps{3}.
\end{align*}
Since $S_j^{\mathbf{X}+\varepsilon\mathbf{Z}}=S_j+\varepsilon S_j^\mathbf{Z}$, and assuming everything is well-defined so that we can use Leibniz integral rule, then we have the following results
\begin{align*}
&\dfrac{\partial \, }{\partial \varepsilon}\mathbb{E}\left(e^{-\frac{S_j^{\mathbf{X}+\varepsilon\mathbf{Z}}}{\beta_j}}\right)=\mathbb{E}\left(-\frac{S_j^\mathbf{Z}}{\beta_j}e^{-\frac{S_j^{\mathbf{X}+\varepsilon\mathbf{Z}}}{\beta_j}}\right);
\\
&\dfrac{\partial}{\partial \varepsilon}\mathbb{E}\left((S_j+\varepsilon S_j^\mathbf{Z})\,e^{-\frac{S_j^{\mathbf{X}+\varepsilon\mathbf{Z}}}{\beta_j}}\right)=\mathbb{E}\left[S_j^\mathbf{Z}\,e^{-\frac{S_j^{\mathbf{X}+\varepsilon\mathbf{Z}}}{\beta_j}}-(S_j+\varepsilon S_j^\mathbf{Z})\frac{S_j^\mathbf{Z}}{\beta_j} e^{-\frac{S_j^{\mathbf{X}+\varepsilon\mathbf{Z}}}{\beta_j}}\right].
\end{align*}

\noindent Compute the derivatives,
\begin{align*}
\dfrac{\partial \, \RomanNumeralCaps{1}}{\partial \varepsilon}
&=-w_{i,j}\dfrac{\partial}{\partial \varepsilon}\left(\dfrac{\mathbb{E}\left((X^i+\varepsilon Z^i)\exp\left({-\frac{S_j^{\mathbf{X}+\varepsilon \mathbf{Z}}}{\beta_j}}\right)\right)}{\mathbb{E}\left(\exp\left({-\frac{S_j^{\mathbf{X}+\varepsilon \mathbf{Z}}}{\beta_j}}\right)\right)}\right)
\\
&=-w_{i,j}\,\dfrac{1}{\left(\mathbb{E} e^{-\frac{S_j^{\mathbf{X}+\varepsilon\mathbf{Z}}}{\beta_j}}\right)^2}
\bigg[\mathbb{E}\left(Z^i e^{-\frac{S_j^{\mathbf{X}+\varepsilon\mathbf{Z}}}{\beta_j}}\right)\,\mathbb{E}\left(e^{-\frac{S_j^{\mathbf{X}+\varepsilon\mathbf{Z}}}{\beta_j}}\right)\\
&
+\mathbb{E}\left((X^i+\varepsilon Z^i)(-\frac{S_j^\mathbf{Z}}{\beta_j})\,e^{-\frac{S_j^{\mathbf{X}+\varepsilon\mathbf{Z}}}{\beta_j}}\right) \mathbb{E}\left(e^{-\frac{S_j^{\mathbf{X}+\varepsilon\mathbf{Z}}}{\beta_j}}\right)
-\mathbb{E}\left((X^i+\varepsilon Z^i)\,e^{-\frac{S_j^{\mathbf{X}+\varepsilon\mathbf{Z}}}{\beta_j}}\right) \mathbb{E}\left(-\frac{S_j^\mathbf{Z}}{\beta_j}\,e^{-\frac{S_j^{\mathbf{X}+\varepsilon\mathbf{Z}}}{\beta_j}}\right)\bigg]
\\
&=-w_{i,j}\left[\mathbb{E}_{\mathbb{Q}^j_{\mathbf{X}+\varepsilon\mathbf{Z}}}[Z^i]-\frac{1}{\beta_j}\mathbb{E}_{\mathbb{Q}^j_{\mathbf{X}+\varepsilon\mathbf{Z}}}\left[(X^i+\varepsilon Z^i)S_j^\mathbf{Z}\right]+\frac{1}{\beta_j}\mathbb{E}_{\mathbb{Q}^j_{\mathbf{X}+\varepsilon\mathbf{Z}}}\left[X^i+\varepsilon Z^i\right]\mathbb{E}_{\mathbb{Q}^j_{\mathbf{X}+\varepsilon\mathbf{Z}}}[S_j^\mathbf{Z}]\right]
\\
&=\mathbb{E}_{\mathbb{Q}^j_{\mathbf{X}+\varepsilon\mathbf{Z}}}\left[-w_{i,j}Z^i\right]+\dfrac{w_{i,j}}{\beta_j}\Cov_{\mathbb{Q}^j_{\mathbf{X}+\varepsilon\mathbf{Z}}}\left(X^i+\varepsilon Z^i,S_j^\mathbf{Z}\right);
\end{align*}
\begin{align*}
\dfrac{\partial \, \RomanNumeralCaps{2}}{\partial \varepsilon}
&=\dfrac{1}{\alpha_i\beta_j}\,\dfrac{\partial}{\partial \varepsilon}\left(\dfrac{\mathbb{E}\left((S_j+\varepsilon S_j^\mathbf{Z})\exp\left({-\frac{S_j^{\mathbf{X}+\varepsilon \mathbf{Z}}}{\beta_j}}\right)\right)}{\mathbb{E}\left(\exp\left({-\frac{S_j^{\mathbf{X}+\varepsilon \mathbf{Z}}}{\beta_j}}\right)\right)}\right)
\\
&=\dfrac{1}{\alpha_i\beta_j}\dfrac{1}{\left(\mathbb{E}e^{-\frac{S_j^{\mathbf{X}+\varepsilon\mathbf{Z}}}{\beta_j}}\right)^2}\,
\bigg[
\left(\mathbb{E}\left(S_j^\mathbf{Z}\,e^{-\frac{S_j^{\mathbf{X}+\varepsilon\mathbf{Z}}}{\beta_j}}\right)+\mathbb{E}\left( (S_j+\varepsilon S_j^\mathbf{Z})(-\frac{S_j^\mathbf{Z}}{\beta_j})e^{-\frac{S_j^{\mathbf{X}+\varepsilon\mathbf{Z}}}{\beta_j}}\right)\right) \mathbb{E}\left(e^{-\frac{S_j^{\mathbf{X}+\varepsilon\mathbf{Z}}}{\beta_j}}\right)\\
&-\mathbb{E}\left((S_j+\varepsilon S_j^\mathbf{Z})\,e^{-\frac{S_j^{\mathbf{X}+\varepsilon\mathbf{Z}}}{\beta_j}}\right) \mathbb{E}\left(-\frac{S_j^\mathbf{Z}}{\beta_j} e^{-\frac{S_j^{\mathbf{X}+\varepsilon\mathbf{Z}}}{\beta_j}}\right)
\bigg]
\\
&=\dfrac{1}{\alpha_i\beta_j}\mathbb{E}_{\mathbb{Q}^j_{\mathbf{X}+\varepsilon\mathbf{Z}}}\left[S_j^\mathbf{Z}\right]-\dfrac{1}{\alpha_i\beta_j^2} \Cov_{\mathbb{Q}^j_{\mathbf{X}+\varepsilon\mathbf{Z}}}\left(S_j+\varepsilon S_j^\mathbf{Z}, S_j^\mathbf{Z}\right);
\\
\dfrac{\partial \, \RomanNumeralCaps{3}}{\partial \varepsilon}
&=\dfrac{1}{\alpha_i}\dfrac{1}{\mathbb{E}\left(e^{-\frac{S_j^{\mathbf{X}+\varepsilon\mathbf{Z}}}{\beta_j}}\right)}\cdot \mathbb{E}\left(-\dfrac{S_j^\mathbf{Z}}{\beta_j}\, e^{-\frac{S_j^{\mathbf{X}+\varepsilon\mathbf{Z}}}{\beta_j}}\right)=-\dfrac{1}{\alpha_i\beta_j}\mathbb{E}_{\mathbb{Q}^j_{\mathbf{X}+\varepsilon\mathbf{Z}}}\left[S_j^\mathbf{Z}\right].
\end{align*}

\noindent As a result
\begin{align*}
\dfrac{\partial \,\mathbb{E}_{\mathbb{Q}^j_{\mathbf{X}+\varepsilon\mathbf{Z}}}\left[Y^{i,j}_{\mathbf{X}+\varepsilon\mathbf{Z}}\right]}{\partial \varepsilon} &=\dfrac{\partial}{\partial \varepsilon}\left( \RomanNumeralCaps{1}+ \RomanNumeralCaps{2}+ \RomanNumeralCaps{3}\right)
\\
&=-w_{i,j}\mathbb{E}_{\mathbb{Q}^j_{\mathbf{X}+\varepsilon\mathbf{Z}}}\left[Z^i\right]+\frac{w_{i,j}}{\beta_j}\Cov_{\mathbb{Q}^j_{\mathbf{X}+\varepsilon\mathbf{Z}}}\left(X^i+\varepsilon Z^i,S_j^\mathbf{Z}\right)+
\dfrac{1}{\alpha_i\beta_j}\mathbb{E}_{\mathbb{Q}^j_{\mathbf{X}+\varepsilon\mathbf{Z}}}\left[S_j^\mathbf{Z}\right]\\
&-\dfrac{1}{\alpha_i\beta_j^2} \Cov_{\mathbb{Q}^j_{\mathbf{X}+\varepsilon\mathbf{Z}}}\left(S_j+\varepsilon S_j^\mathbf{Z}, S_j^\mathbf{Z}\right)-\dfrac{1}{\alpha_i\beta_j}\mathbb{E}_{\mathbb{Q}^j_{\mathbf{X}+\varepsilon\mathbf{Z}}}\left[S_j^\mathbf{Z}\right]
\\
&=\mathbb{E}_{\mathbb{Q}^j_{\mathbf{X}+\varepsilon\mathbf{Z}}}\left[-w_{i,j}Z^i\right]+\frac{w_{i,j}}{\beta_j}\Cov_{\mathbb{Q}^j_{\mathbf{X}+\varepsilon\mathbf{Z}}}\left(X^i+\varepsilon Z^i,S_j^\mathbf{Z}\right)-\dfrac{1}{\alpha_i\beta_j^2} \Cov_{\mathbb{Q}^j_{\mathbf{X}+\varepsilon\mathbf{Z}}}\left(S_j+\varepsilon S_j^\mathbf{Z}, S_j^\mathbf{Z}\right)
\\
&=\mathbb{E}_{\mathbb{Q}^j_{\mathbf{X}+\varepsilon\mathbf{Z}}}\left[-w_{i,j}Z^i\right]-\frac{1}{\beta_j}\Cov_{\mathbb{Q}^j_{\mathbf{X}+\varepsilon\mathbf{Z}}}\left(Y^{i,j}_{\mathbf{X}+\varepsilon\mathbf{Z}}, S_j^\mathbf{Z}\right).
\end{align*}
Then when $\varepsilon=0$, we have the formula for marginal risk allocation for individual $i$ in group $j$ in proposition \ref{proposition: sensitivity analysis-main}.

The marginal risk contribution of group $j$ is trivial and for the conclusion  on local causal responsibility, we have:
\begin{equation*}
    \mathbb{E}_{\mathbb{Q}^j_{\mathbf{X}}}\left[Y^{i,j}_{\mathbf{X}+\varepsilon\mathbf{Z}}\right] =\mathbb{E}_{\mathbb{Q}^j_{\mathbf{X}}}\left[ \left(-w_{i,j}(X^{i}+\varepsilon Z^i)+\frac{1}{\alpha_i\beta_j}\left(S_j+\varepsilon S_j^\mathbf{Z}\right)+\frac{1}{\alpha_i\beta_j}d_j^{\mathbf{X}+\varepsilon \mathbf{Z}}\right)\right].
\end{equation*}
Then according to the previous proof,
\begin{align*}
     \frac{\partial}{\partial \varepsilon}\mathbb{E}_{\mathbb{Q}^j_{\mathbf{X}}}\left[Y^{i,j}_{\mathbf{X}+\varepsilon\mathbf{Z}}\right] 
     &= \mathbb{E}_{\mathbb{Q}^j_{\mathbf{X}}}\left(-w_{i,j}Z^i+\frac{1}{\alpha_i\beta_j}S_j^\mathbf{Z}\right)-\frac{1}{\alpha_i\beta_j}\mathbb{E}_{\mathbb{Q}^j_{\mathbf{X}+\varepsilon\mathbf{Z}}}(S_j^\mathbf{Z}),
\end{align*}
and thus,
\[
\frac{\partial}{\partial \varepsilon}\mathbb{E}_{\mathbb{Q}^j_{\mathbf{X}}}\left[Y^{i,j}_{\mathbf{X}+\varepsilon\mathbf{Z}}\right] \bigg\lvert _{\varepsilon = 0}
     = \mathbb{E}_{\mathbb{Q}^j_{\mathbf{X}}}\left[-w_{i,j}Z^i\right].
\]

\subsection{Proof of Proposition \ref{proposition: sensitivity analysis}} 
\label{Appendix: sensitivity analysis}

By theorem \ref{thm:JointResult}, for $i\in I_j$
\begin{align*}
\mathbb{E}_{\mathbb{Q}_\mathbf{X}^j}\,\left[Y^{i,j}\right]
&=\mathbb{E}\,\left[Y^{i,j}\cdot \dfrac{\mathrm{d}\mathbb{Q}_\mathbf{X}^j}{\mathrm{d}\mathbb{P}}\right]
=\mathbb{E}\,\left[\left(-w_{i,j}X^{i}+\frac{1}{\alpha_i\beta_j}\left(S_j+d_j\right)\right)\cdot\dfrac{e^{-\frac{S_j}{\beta_j}}}{\mathbb{E}\left[e^{-\frac{S_j}{\beta_j}}\right]}\right]
\\
&=-w_{i,j}\,\dfrac{\mathbb{E}\left(X^ie^{-\frac{S_j}{\beta_j}}\right)}{\mathbb{E}\left(e^{-\frac{S_j}{\beta_j}}\right)}
+\frac{1}{\alpha_i\beta_j}\,\dfrac{\mathbb{E}\left(S_je^{-\frac{S_j}{\beta_j}}\right)}{\mathbb{E}\left(e^{-\frac{S_j}{\beta_j}}\right)}
+\frac{1}{\alpha_i}\left(\log (-\beta/B)+\log \left(\mathbb{E}\left(e^{-\frac{S_j}{\beta_j}}\right)\right)\right)
\\
&=\RomanNumeralCaps{1}+\RomanNumeralCaps{2}+\RomanNumeralCaps{3}.
\end{align*}
Assuming everything is well-defined so that we can use Leibniz integral rule, then we have the following results
\begin{align*}
\dfrac{\partial \, }{\partial w_{i,j}}\mathbb{E}\left(e^{-\frac{S_j}{\beta_j}}\right)&=\mathbb{E}\left(-\dfrac{X^i}{\beta_j}e^{-\frac{S_j}{\beta_j}}\right)\quad;\quad
\dfrac{\partial}{\partial w_{i,j}}\mathbb{E}\left(S_j\,e^{-\frac{S_j}{\beta_j}}\right)=\mathbb{E}\left[X^i\,e^{-\frac{S_j}{\beta_j}}-S_j\dfrac{X^i}{\beta_j} e^{-\frac{S_j}{\beta_j}}\right].
\end{align*}

\noindent Compute the derivatives,
\begin{align*}
\dfrac{\partial \, \RomanNumeralCaps{1}}{\partial w_{i,j}}
&=\dfrac{\partial}{\partial w_{i,j}}\left(-w_{i,j } \dfrac{\mathbb{E}\left(X^ie^{-\frac{S_j}{\beta_j}}\right)}{\mathbb{E}\left(e^{-\frac{S_j}{\beta_j}}\right)}\right)
\\
&=-\dfrac{\mathbb{E}\left(X^i\,e^{-\frac{S_j}{\beta_j}}\right)}{\mathbb{E}\left(e^{-\frac{S_j}{\beta_j}}\right)}-w_{i,j}
\dfrac{\mathbb{E}\left(-\frac{(X^i)^2}{\beta_j}\,e^{-\frac{S_j}{\beta_j}}\right)\cdot \mathbb{E}\left(e^{-\frac{S_j}{\beta_j}}\right)-\mathbb{E}\left(X^i\,e^{-\frac{S_j}{\beta_j}}\right)\cdot \mathbb{E}\left(-\frac{X^i}{\beta_j}\,e^{-\frac{S_j}{\beta_j}}\right)}{\left(\mathbb{E}\left(e^{-\frac{S_j}{\beta_j}}\right)\right)^2}
\\
&=-\dfrac{\mathbb{E}\left(X^i\,e^{-\frac{S_j}{\beta_j}}\right)}{\mathbb{E}\left(e^{-\frac{S_j}{\beta_j}}\right)}+\dfrac{w_{i,j}}{\beta_j}\dfrac{\mathbb{E}\left((X^i)^2\,e^{-\frac{S_j}{\beta_j}}\right)}{\mathbb{E}\left(e^{-\frac{S_j}{\beta_j}}\right)}-\dfrac{w_{i,j}}{\beta_j}\left(\dfrac{\mathbb{E}\left(X^i\,e^{-\frac{S_j}{\beta_j}}\right)}{\mathbb{E}\left(e^{-\frac{S_j}{\beta_j}}\right)}\right)^2
\\
&=-\mathbb{E}_{\mathbb{Q}_\mathbf{X}^j}\left[X^i\right]+\dfrac{w_{i,j}}{\beta_j}\Var_{\mathbb{Q}_\mathbf{X}^j}\left(X^i\right);
\end{align*}
\begin{align*}
\dfrac{\partial \, \RomanNumeralCaps{2}}{\partial w_{i,j}}
&=\dfrac{1}{\alpha_i\beta_j}\,\dfrac{\partial}{\partial w_{i,j}}\left(\dfrac{\mathbb{E}\left(S_j\,e^{-\frac{S_j}{\beta_j}}\right)}{\mathbb{E}\left(e^{-\frac{S_j}{\beta_j}}\right)}\right)
\\
&=\dfrac{1}{\alpha_i\beta_j}\dfrac{1}{\left(\mathbb{E}\left(e^{-\frac{S_j}{\beta_j}}\right)\right)^2}\,\left(
\mathbb{E}\left[X^i\,e^{-\frac{S_j}{\beta_j}}-S_j\frac{X^i}{\beta_j}e^{-\frac{S_j}{\beta_j}}\right]\cdot \mathbb{E}\left(e^{-\frac{S_j}{\beta_j}}\right)
-\mathbb{E}\left(S_j\,e^{-\frac{S_j}{\beta_j}}\right)\cdot \mathbb{E}\left(-\frac{X^i}{\beta_j} e^{-\frac{S_j}{\beta_j}}\right)
\right)
\\
&=\dfrac{1}{\alpha_i\beta_j}\left(\dfrac{\mathbb{E}\left(X^i\,e^{-\frac{S_j}{\beta_j}}\right)}{\mathbb{E}\left(e^{-\frac{S_j}{\beta_j}}\right)}
-\dfrac{1}{\beta_j}\mathbb{E}\left(X^iS_j\,\dfrac{e^{-\frac{S_j}{\beta_j}}}{\mathbb{E}\left(e^{-\frac{S_j}{\beta_j}}\right)}\right)+\dfrac{1}{\beta_j}\mathbb{E}\left(X^i\,\dfrac{e^{-\frac{S_j}{\beta_j}}}{\mathbb{E}\left(e^{-\frac{S_j}{\beta_j}}\right)}\right)\cdot \mathbb{E}\left(S_j\,\dfrac{e^{-\frac{S_j}{\beta_j}}}{\mathbb{E}\left(e^{-\frac{S_j}{\beta_j}}\right)}\right)
\right)
\\
&=\dfrac{1}{\alpha_i\beta_j}\mathbb{E}_{\mathbb{Q}_\mathbf{X}^j}\left[X^i\right]-\dfrac{1}{\alpha_i\beta_j^2} \Cov_{\mathbb{Q}_\mathbf{X}^j}\left(X^i, S_j\right);
\end{align*}
\begin{align*}
\dfrac{\partial \, \RomanNumeralCaps{3}}{\partial w_{i,j}}
&=\dfrac{\partial}{\partial w_{i,j}}\left( \frac{1}{\alpha_i}\log (-\beta/B)+\frac{1}{\alpha_i}\log \left(\mathbb{E}\left(e^{-\frac{S_j}{\beta_j}}\right)\right)\right)
=\dfrac{1}{\alpha_i}\dfrac{1}{\mathbb{E}\left(e^{-\frac{S_j}{\beta_j}}\right)}\cdot \mathbb{E}\left(-\dfrac{X^i}{\beta_j}\, e^{-\frac{S_j}{\beta_j}}\right)
\\
&=-\dfrac{1}{\alpha_i\beta_j}\mathbb{E}_{\mathbb{Q}_\mathbf{X}^j}\left[X^i\right].
\end{align*}

\noindent As a result
\begin{align*}
\dfrac{\partial \,\mathbb{E}_{\mathbb{Q}_\mathbf{X}^j}\left[Y^{i,j}\right]}{\partial w_{i,j}}&=\dfrac{\partial}{\partial w_{i,j}}\left( \RomanNumeralCaps{1}+ \RomanNumeralCaps{2}+ \RomanNumeralCaps{3}\right)
\\
&=-\mathbb{E}_{\mathbb{Q}_\mathbf{X}^j}\left[X^i\right]+\dfrac{w_{i,j}}{\beta_j}\Var_{\mathbb{Q}_\mathbf{X}^j}\left(X^i\right)+\dfrac{1}{\alpha_i\beta_j}\mathbb{E}_{\mathbb{Q}_\mathbf{X}^j}\left[X^i\right]-\dfrac{1}{\alpha_i\beta_j^2} \Cov_{\mathbb{Q}_\mathbf{X}^j}\left(X^i, S_j\right)-\dfrac{1}{\alpha_i\beta_j}\mathbb{E}_{\mathbb{Q}_\mathbf{X}^j}\left[X^i\right]
\\
&=-\mathbb{E}_{\mathbb{Q}_\mathbf{X}^j}\left[X^i\right]-\dfrac{1}{\alpha_i\beta_j^2} \Cov_{\mathbb{Q}_\mathbf{X}^j}\left(X^i, S_j\right)+\dfrac{w_{i,j}}{\beta_j}\Var_{\mathbb{Q}_\mathbf{X}^j}\left(X^i\right),\quad i\in I_j.
\end{align*}

\subsection{Proof of Proposition \ref{prop: monotonicity-2}}
\label{appendix: monotonicity-2}
 Define 
 \[ \eta_m'=\sum\limits_{k\in I_{m'}}\dfrac{w_{k,m'}}{w_{k,m}}\dfrac{1}{\alpha_k}.
\]
By theorem \ref{thm:JointResult}, for $k\in I_m,$
\begin{align*}
\sum\limits_{k\in I_{m'}}\dfrac{w_{k,m'}}{w_{k,m}}Y^{k,m}
&=\dfrac{S_m+d_m}{\beta_m}\sum\limits_{k\in I_{m'}}\dfrac{w_{k,m'}}{w_{k,m}}\dfrac{1}{\alpha_k}-\sum\limits_{k\in I_{m'}}\dfrac{w_{k,m'}}{w_{k,m}}w_{k,m}X^k\nonumber
\\
&=\dfrac{S_m+d_m}{\beta_m}\eta_m'-\sum\limits_{k\in I_{m'}}w_{k,m'}X^k.
\end{align*}
Then
\begin{align}
\mathbb{E}_{{\mathbb Q}_\mathbf{X}^m}&\left[\sum\limits_{k\in I_{m'}}\dfrac{w_{k,m'}}{w_{k,m}}Y^{k,m}\right]\nonumber\\
&=\mathbb{E}_{{\mathbb Q}_\mathbf{X}^m}\left(\dfrac{\eta_m'}{\beta_m}S_m-\sum\limits_{k\in I_{m'}}w_{k,m'}X^k\right)+\dfrac{\eta_m'}{\beta_m} \beta_m\log\left\{-\dfrac{\beta}{B}\mathbb{E}\left[\exp\left(-\frac{S_m}{\beta_m} \right)\right]\right\}\nonumber
\\
&=\eta_m'\log\left\{\exp\left[\dfrac{1}{\eta_m'} \mathbb{E}_{{\mathbb Q}_\mathbf{X}^m}\left(\dfrac{\eta_m'}{\beta_m}S_m-\sum\limits_{k\in I_{m'}}w_{k,m'}X^k\right)\right]\right\}+\eta_m'\log\left\{-\dfrac{\beta}{B}\mathbb{E}\left[\exp\left(-\frac{S_m}{\beta_m} \right)\right]\right\}\nonumber
\\
&\leq \eta_m'\log \left\{\mathbb{E}_{{\mathbb Q}_\mathbf{X}^m}\left[\exp\left(\dfrac{1}{\beta_m}S_m-\dfrac{1}{\eta_m'}\sum\limits_{k\in I_{m'}}w_{k,m'}X^k\right)\right]\right\}+\eta_m'\log\left\{-\dfrac{\beta}{B}\mathbb{E}\left[\exp\left(-\frac{S_m}{\beta_m} \right)\right]\right\}\nonumber
\\
&=\eta_m'\log \left\{\mathbb{E}\left[
\dfrac{e^{\frac{S_m}{\beta_m}}\,e^{-\frac{S_m}{\beta_m} }\,\exp\left(-\dfrac{1}{\eta_m'}\sum\limits_{k\in I_{m'}}w_{k,m'}X^k\right)}
{\mathbb{E}[e^{-\frac{S_m}{\beta_m} }]}
\right]\right\}
 +\eta_m'\log\left\{-\dfrac{\beta}{B}\mathbb{E}\left[\exp\left(-\frac{S_m}{\beta_m} \right)\right]\right\}\nonumber
\\
&=\eta_m'\log\left\{-\dfrac{\beta}{B}\mathbb{E}\left[\exp\left(-\frac{1}{\eta_m'}\sum\limits_{k\in I_{m'}}w_{k,m'}X^k \right)\right]\right\}\label{eqn: monotonicity-2}
\\
&<\beta_m'\log\left\{-\dfrac{\beta'}{B}\mathbb{E}\left[\exp\left(-\frac{1}{\beta_m'}\sum\limits_{k\in I_{m'}}w_{k,m'}X^k \right)\right]\right\},\quad\text{if }\sum\limits_{k\in I_{m'}}w_{k,m'}X^k\text{ is nonnegative,}\nonumber
\\
&:= d_{m'}.\nonumber
\end{align}

In conclusion, if both $\sum\limits_{k\in I_{m'}}w_{k,m'}X^k$ and $\sum\limits_{k\in I_{m''}}w_{k,m''}X^k$ are nonnegative, the inequality holds for the risk allocations of subgroup $I_{m'}$, as well as $I_{m''}.$ Otherwise, we have the inequality given by \eqref{eqn: monotonicity-2}.

\subsection{Necessary and Sufficient Condition for \texorpdfstring{$B$}{B} in Remark \ref{remark: Necessary and Sufficient Condition for B}}
\label{appendix: Necessary and Sufficient Condition for B}
Here we show a necessary and sufficient condition for $B$ to have trivial Nash through an example. We assume all risk factors are i.i.d. Gaussian random variables where $\sigma_{ij}=0$, $i\neq j$ and $\sigma_{ii}=\sigma$ for all $i,j$.

When all banks in one group, {\it i.e.}, $h=1$, $\beta = \sum\limits_{i=1}^N\frac{1}{\alpha_i}$,
\begin{align}
\mathbb E _{\mathbb Q_{\mathbf X}} [ Y^{1}_{\mathbf X} ] = \mathbb E _{\mathbb Q^{1}_{\mathbf X}} [ Y^{1,1}_{\mathbf X} ]
&=\frac{1}{\alpha_1}\log(\frac{\beta}{-B})-\mu_1+\frac{1}{\beta_1}\sigma -\frac{1}{2\beta_1^2\alpha_1} N\sigma.\label{eqn: trivialNashEx1}
\end{align}
When bank 1 decides to split and put some weights in another group, e.g. $\exists\, w_{1,1},w_{1,2}>0$ and $w_{1,1}+w_{1,2}=1$, then $\beta'=\beta + \frac{1}{\alpha_1}$ and $\beta_1=\beta$, $\beta_2=\frac{1}{\alpha_1}$,
\begin{align}
\mathbb E _{\mathbb Q_{\mathbf X}} [ Y^{1}_{\mathbf X} ] &= \mathbb E _{\mathbb Q^{1}_{\mathbf X}} [ Y^{1,1}_{\mathbf X} ]+\mathbb E _{\mathbb Q^{2}_{\mathbf X}} [ Y^{1,2}_{\mathbf X} ]\nonumber\\
&=\frac{2}{\alpha_1}\log(\frac{\beta'}{-B})-\mu_1+\frac{w_{1,1}^2}{\beta_1}\sigma+\frac{w_{1,2}^2}{\beta_2}\sigma\nonumber
\\
&-\frac{1}{2\beta_1^2\alpha_1} \left((N-1)\sigma +w_{1,1}^2\sigma\right)
-\frac{1}{2\beta_2^2\alpha_1} w_{1,2}^2\sigma.\label{eqn: trivialNashEx2}
\end{align}
To have trivial Nash, for bank 1, it should hold that \eqref{eqn: trivialNashEx2}$\geq$\eqref{eqn: trivialNashEx1}, which gives:
\begin{align}
    \frac{1}{\alpha_1}\log (-B)\leq \frac{1}{\alpha_1}\log\left(\frac{(\beta')^2}{\beta}\right)-\left[\frac{1}{\beta_1}-\left(\frac{w_{1,1}^2}{\beta_1}+\frac{w_{1,2}^2}{\beta_2}\right)\right]\sigma-\frac{1}{2\alpha_1}\left[\left(\frac{w_{1,1}^2}{\beta_1^2}+\frac{w_{1,2}^2}{\beta_2^2}\right)-\frac{1}{\beta_1^2}\right]\sigma.\label{eqn: trivialNashEx3}
\end{align}
Then by extending \eqref{eqn: trivialNashEx3} to all banks,
we can get the necessary and sufficient condition on $B$ to have trivial Nash: for all $i=1,\ldots,N$, $B$ satisfies
\begin{align}
    \frac{1}{\alpha_i}\log (-B)\leq \frac{1}{\alpha_i}\log\left(\frac{(\beta')^2}{\beta}\right)-\left[\frac{1}{\beta_1}-\left(\frac{w_{i,1}^2}{\beta_1}+\frac{w_{i,2}^2}{\beta_2}\right)\right]\sigma-\frac{1}{2\alpha_i}\left[\left(\frac{w_{i,1}^2}{\beta_1^2}+\frac{w_{i,2}^2}{\beta_2^2}\right)-\frac{1}{\beta_1^2}\right]\sigma,
\end{align}
where $\beta_2=\frac{1}{\alpha_i}$, $\beta_1=\sum\limits_{i=1}^N\frac{1}{\alpha_i}$ and $\beta' = \beta_1+\beta_2$.

Recall that $B$ is negative and stands for the minimal level of expected utility. Intuitively, when $B$ is small, $\log (-B)$ is large, then some of inequalities tend to be violated so that there would be no trivial Nash in the system. On the other hand, when $B$ is large (close to 0), $\log (-B)$ will be extremely small so that trivial Nash may exist in the system. 

\subsection{Remark for \texorpdfstring{$S_m$}{Sm} in equation \texorpdfstring{\eqref{eqn: Sm distribution}}{}}
\label{proof: expResults}
\begin{remark}
Some results about $S_m$:
\begin{itemize}
\item $\mathbb{E}\big(e^{-S_m/\beta_m}\big)=\exp\big(-\frac{1}{\beta_m}\,\mu_m^s\,+\,\frac{1}{2\beta_m^2} (\sigma_m^s)^2 \big);$
\item For $i\in I_m$, (proof see below)
\[
\mathbb{E}\big(X^ie^{-S_m/\beta_m}\big)=\big(\mu_i-\frac{1}{\beta_m}A_m{\bm \Sigma}_{[,i]}  \big)\exp\big(-\frac{1}{\beta_m}\,\mu_m^s\,+\,\frac{1}{2\beta_m^2} (\sigma_m^s)^2 \big);
\]
\item $
\mathbb{E}\big(S_me^{-S_m/\beta_m}\big)=\big(\mu_m^s-\frac{1}{\beta_m}(\sigma_m^s)^2  \big)\exp\big(-\frac{1}{\beta_m}\,\mu_m^s\,+\,\frac{1}{2\beta_m^2} (\sigma_m^s)^2 \big);$
\end{itemize}
\end{remark}
\begin{proof}
Define $\mathbf{t}^T=(t_1,\cdots,t_N)$
\begin{align}
\mathbb{E}\big(X^ie^{\mathbf{t}^T\mathbf{X}}\big)
&=\dfrac{\partial M_\mathbf{x} (\mathbf{t})}{\partial t_i}=
\dfrac{\partial}{\partial t_i}\,\exp\big({\bm \mu}^T\mathbf{t}+\frac{1}{2}\mathbf{t}^T{\bm \Sigma}\mathbf{t}  \big)\\
&=\big(\mu_i+\mathbf{t}^T{\bm \Sigma}_{[,i]}\big)\,\exp\big({\bm \mu}^T\mathbf{t}+\frac{1}{2}\mathbf{t}^T{\bm \Sigma}\mathbf{t}\big)
\end{align}
So 
\[
\mathbb{E}\big(X^ie^{-S_m/\beta_m}\big)=\mathbb{E}\big(X^ie^{-(A_m/\beta_m) \mathbf{X}}\big)=\big(\mu_i-\frac{1}{\beta_m}A_m{\bm \Sigma}_{[,i]}  \big)\exp\big(-\frac{1}{\beta_m}\,\mu_m^s\,+\,\frac{1}{2\beta_m^2} (\sigma_m^s)^2 \big).
\]
\end{proof}

\subsection{Sufficient Condition for Local Optimal Weights}
\label{appendix: Sufficient Condition for Local Optimal Weights}
To investigate the condition \eqref{eqn:w^*condition} further, we make some reasonable assumptions on estimates and introduce some situations when they hold.
\begin{itemize}
\item assuming: all $w_{k,1},w_{k,2}\neq 0$, for all $k\neq i$, then $\beta_1 = \beta_2=\sum\frac{1}{\alpha_k}.$

Then \eqref{eqn:w^*condition} is equivalent to 
\begin{equation}
-(\frac{2}{\beta_1}-\frac{1}{\beta_1^2\alpha_i})\sigma_{ii}<\left(\frac{1}{\beta_1^2\alpha_i}-\frac{1}{\beta_1}\right)\sum\limits_{k=1,k\neq i}^N(w_{k,1}-w_{k,2})\sigma_{ki} < (\frac{2}{\beta_1}-\frac{1}{\beta_1^2\alpha_i})\sigma_{ii}.
\end{equation}
Since 
\[
\frac{1}{\beta_1^2\alpha_i}-\frac{1}{\beta_1}<0,\quad\quad 
|(w_{k,1}-w_{k,2})\sigma_{ki}|\leq |\sigma_{ki}| \quad \text{ for all }k,i,
\]
we can deduce a {sufficient} condition for $w_{(i)}^*\in (0,1)$ for individual $i$:
\begin{equation}\label{eqn:w^*condition-sufficient-1}
\left(1-\frac{1}{\beta_1\alpha_i}\right)\sum\limits_{k=1,k\neq i}^N |\sigma_{ki}| < (2-\frac{1}{\beta_1\alpha_i})\sigma_{ii}.
\end{equation} 
\begin{remark}
From above, we can have a rough estimation: if $\sigma_i\sim \sigma$, $\rho >0$ and $\alpha = [1,1,1,1]$ ({\it i.e.}, {no extremely large $\sigma$ and no extremely small $\alpha$}), the inequality is true when $\rho < 7/9$, according to
\[
(1-\frac{1}{4})\cdot (4-1)\rho\sigma^2 \,<\, (2-\frac{1}{4})\sigma^2.
\]
This explains why in numerical experiments, when we apply reasonable values of parameters, the optimal weights are often located between $(0,1)$. And when $(w_{k,1},w_{k,2})$ are close for most $k\neq i$, the weights for individual $i$ are around $0.5$ because of small $A\sim 0$ and $B_1\sim B_2$ in this case.

If for individual $i$, $\sigma_i$ is small and $\sigma_{i}<<\sigma_{i^0}$ for some $i^0$, then the sufficient condition doesn't hold and by numerical results, we found $w^*_{(i)}$ is not in $(0,1)$ anymore.

\end{remark}

Assuming for some $k$('s), $w_{k,1}, w_{k,2}$ can be $0$ or $1$. Then $\beta_1\neq \beta_2$ but it is still true that $-1\leq w_{k,1}-w_{k,2}\leq 1$ for all $k$.

\eqref{eqn:w^*condition} can be rewritten as 
\begin{equation}
-(\frac{2}{\beta_2}-\frac{1}{\beta_2^2\alpha_i})\sigma_{ii}<\sum\limits_{k=1,k\neq i}^N\left(w_{k,1}(\frac{1}{\beta_1^2\alpha_i}-\frac{1}{\beta_1})-w_{k,2}(\frac{1}{\beta_2^2\alpha_i}-\frac{1}{\beta_2})\right)\sigma_{ki} < (\frac{2}{\beta_1}-\frac{1}{\beta_1^2\alpha_i})\sigma_{ii}.
\end{equation}
A {sufficient} condition for $w_{(i)}^*\in (0,1)$ for individual $i$:
\begin{equation}
\left\{
\begin{array}{ll}
&\left(\frac{1}{\beta_1}-\frac{1}{\beta_1^2\alpha_i}\right)\sum\limits_{k=1,k\neq i}^N |\sigma_{ki}| < (\frac{2}{\beta_2}-\frac{1}{\beta_2^2\alpha_i})\sigma_{ii},
\\
&\left(\frac{1}{\beta_2}-\frac{1}{\beta_2^2\alpha_i}\right)\sum\limits_{k=1,k\neq i}^N |\sigma_{ki}| < (\frac{2}{\beta_1}-\frac{1}{\beta_1^2\alpha_i})\sigma_{ii}.
\end{array}
\right.
\end{equation} 
This is a generalization of \eqref{eqn:w^*condition-sufficient-1}.
\end{itemize}

\subsection{Necessary and Sufficient Condition for Optimal Weights}
\label{appendix: Necessary and Sufficient Condition for Optimal Weights}
First, we compare the minimal risk over non-zero weights \eqref{eqn:non-corner optimal-risk} with the corner case $(w_{i1},w_{i2})=(0,1)$:
\begin{align}
\mathbb{E}_{\mathbb{Q}_\mathbf{X}}\left[Y^{i}\right]\bigg|_{w=w^*}-\eqref{eqn:corner-2} 
&=\frac{2}{\alpha_i}\log(\frac{\beta'}{-B})-\frac{1}{\alpha_i}\log(\frac{\beta}{-B})\\ &+w^* \sum\limits_{k=1,k\neq i}^N \left[\left(\frac{w_{k,1}}{\beta_1}-\frac{w_{k,2}}{\beta_2}\right)-\left(\frac{w_{k,1}}{\beta_1^2\alpha_i}-\frac{w_{k,2}}{\beta_2^2\alpha_i}\right)\right]\sigma_{ki}\nonumber
\\
&+\left(\frac{(w^*)^2}{\beta_1}-\frac{(w^*)^2}{2\beta_1^2\alpha_i}\right)\sigma_{ii}+((1-w^*)^2-1)\left(\frac{1}{\beta_2}-\frac{1}{2\beta_2^2\alpha_i}\right)\sigma_{ii}\nonumber
\\
&-\frac{1}{2\beta_1^2\alpha_i} 
\sum\limits_{m,k=1,\neq i}^N w_{k,1}w_{m,1}\sigma_{km}\nonumber
\\
(\text{use notation $A,B_1,B_2$})\quad
&=\frac{2}{\alpha_i}\log(\frac{\beta'}{-B})-\frac{1}{\alpha_i}\log(\frac{\beta}{-B})+w^*\cdot (-A)+\frac{(w^*)^2}{2}B_1+\frac{w^*(w^*-2)}{2}B_2\nonumber
\\
&-\frac{1}{2\beta_1^2\alpha_i} 
\sum\limits_{m,k=1,\neq i}^N w_{k,1}w_{m,1}\sigma_{km}\nonumber
\\
&=\frac{2}{\alpha_i}\log(\frac{\beta'}{-B})-\frac{1}{\alpha_i}\log(\frac{\beta}{-B})-w^*\cdot \underbrace{(A-\frac{w^*}{2}B_1-\frac{w^*-2}{2}B_2)}_{=\Delta}\nonumber
\\
&-\frac{1}{2\beta_1^2\alpha_i} 
\sum\limits_{m,k=1,\neq i}^N w_{k,1}w_{m,1}\sigma_{km}\nonumber
\\
&=\frac{2}{\alpha_i}\log(\frac{\beta'}{-B})-\frac{1}{\alpha_i}\log(\frac{\beta}{-B})-\dfrac{(A+B_2)^2}{2(B_1+B_2)}
-\frac{1}{2\beta_1^2\alpha_i} 
\sum\limits_{m,k=1,\neq i}^N w_{k,1}w_{m,1}\sigma_{km}.\label{eqn:condition-2}
\end{align}
Since $w^*=\frac{A+B_2}{B_1+B_2}$, we get $\Delta=A-\frac{w^*}{2}(B_1+B_2)+B_2=\frac{A+B_2}{2}$.

Then we compare the minimal risk of non-boundary case \eqref{eqn:non-corner optimal-risk} with the boundary case $(w_{i1},w_{i2})=(1,0)$:
\begin{align}
\mathbb{E}_{\mathbb{Q}_\mathbf{X}}\left[Y^{i}\right]\bigg|_{w=w^*}-\eqref{eqn:corner-1} 
&=\frac{2}{\alpha_i}\log(\frac{\beta'}{-B})-\frac{1}{\alpha_i}\log(\frac{\beta}{-B})\nonumber
\\
&+(w^*-1) \sum\limits_{k=1,k\neq i}^N \left[\left(\frac{w_{k,1}}{\beta_1}-\frac{w_{k,2}}{\beta_2}\right)-\left(\frac{w_{k,1}}{\beta_1^2\alpha_i}-\frac{w_{k,2}}{\beta_2^2\alpha_i}\right)\right]\sigma_{ki}\nonumber
\\
&+((w^*)^2-1)\left(\frac{1}{\beta_1}-\frac{1}{2\beta_1^2\alpha_i}\right)\sigma_{ii}+(1-w^*)^2\left(\frac{1}{\beta_2}-\frac{1}{2\beta_2^2\alpha_i}\right)\sigma_{ii}\nonumber
\\
&-\frac{1}{2\beta_2^2\alpha_i} 
\sum\limits_{m,k=1,\neq i}^N w_{k,2}w_{m,2}\sigma_{km}\nonumber
\\
&=\frac{2}{\alpha_i}\log(\frac{\beta'}{-B})-\frac{1}{\alpha_i}\log(\frac{\beta}{-B})-(w^*-1)\left(A-\frac{w^*+1}{2}B_1-\frac{w^*-1}{2}B_2\right)\nonumber
\\
&-\frac{1}{2\beta_2^2\alpha_i} 
\sum\limits_{m,k=1,\neq i}^N w_{k,2}w_{m,2}\sigma_{km}\nonumber
\\
&=\frac{2}{\alpha_i}\log(\frac{\beta'}{-B})-\frac{1}{\alpha_i}\log(\frac{\beta}{-B})-\dfrac{(A-B_1)^2}{2(B_1+B_2)}
-\frac{1}{2\beta_2^2\alpha_i} 
\sum\limits_{m,k=1,\neq i}^N w_{k,2}w_{m,2}\sigma_{km}.\label{eqn:condition-1}
\end{align}

If both \eqref{eqn:condition-1} and \eqref{eqn:condition-2} are less than $0$, we get conditions \eqref{eqn:condition}, which are the necessary and sufficient conditions to conclude non-zero weights $(w^*,1-w^*)$ are the optimal weights to minimize the total risk.

\end{document}